\documentclass[11pt]{iopart}
\pdfoutput=1 
\usepackage[T1]{fontenc}

\expandafter\let\csname equation*\endcsname\relax
\expandafter\let\csname endequation*\endcsname\relax
\usepackage{amsmath,amssymb,mathtools} 
\usepackage{amsthm}

\usepackage[shortlabels]{enumitem}

\usepackage{mathabx}

\usepackage{tikz}
\usetikzlibrary{positioning, calc}
\usetikzlibrary{calc}
\usetikzlibrary{arrows}
\usepackage{tikz-3dplot}
\usepackage{graphicx}
\usepackage{mdwlist}
\usepackage{color}
\usepackage[top=2.5cm,bottom=2.5cm,left=2.5cm,right=2.5cm]{geometry}

\usetikzlibrary{decorations.pathreplacing,decorations.markings}
\usetikzlibrary{decorations.pathmorphing}

\tikzset{snake it/.style={decorate, decoration=snake}}

\usepackage{tcolorbox}
\newtcolorbox{myredbox}[1]{colback=red!5!white,colframe=red!75!black,fonttitle=\bfseries,title=#1}

\makeatletter
\def\@mkboth#1#2{}
\newlength\appendixwidth
\preto\appendix{\addtocontents{toc}{\protect\patchl@section}}
\newcommand{\patchl@section}{%
  \settowidth{\appendixwidth}{\textbf{Appendix }}%
  \addtolength{\appendixwidth}{1.5em}%
  \patchcmd{\l@section}{1.5em}{\appendixwidth}{}{\ddt}%
}
\newrobustcmd{\fixappendix}{%
  \patchcmd{\l@section}{1.5em}{7em}{}{}%
  \patchcmd{\l@subsection}{2.3em}{7em}{}{}%
}
\makeatother

\usepackage{footmisc}

\tikzset{
    >=stealth',
    punkt/.style={
           rectangle,
           rounded corners,
           draw=black, very thick,
           text width=6.5em,
           minimum height=2em,
           text centered},
    pil/.style={
           -{Latex[length=3mm,width=10mm]},
           thick,
           shorten <=2pt,
           shorten >=2pt,},
  on each segment/.style={
    decorate,
    decoration={
      show path construction,
      moveto code={},
      lineto code={
        \path [#1]
        (\tikzinputsegmentfirst) -- (\tikzinputsegmentlast);
      },
      curveto code={
        \path [#1] (\tikzinputsegmentfirst)
        .. controls
        (\tikzinputsegmentsupporta) and (\tikzinputsegmentsupportb)
        ..
        (\tikzinputsegmentlast);
      },
      closepath code={
        \path [#1]
        (\tikzinputsegmentfirst) -- (\tikzinputsegmentlast);
      },
    },
  },
  mid arrow/.style={postaction={decorate,decoration={
        markings,
        mark=at position .5 with {\arrow[#1]{stealth'}}
      }}}
}

\usepackage[hyperindex,breaklinks]{hyperref}
\hypersetup{
    colorlinks=true,
    linkcolor=blue,
    filecolor=black,      
    urlcolor=black,
    citecolor=black
}
\usepackage{slashed}

\let\oldmarginpar\marginpar
\renewcommand\marginpar[1]{\-\oldmarginpar[\raggedleft\footnotesize #1]%
{\raggedright\footnotesize #1}}

\usepackage{float}		
\usepackage{graphicx}		
\usepackage{subcaption}

 \newcommand{\ket}[1]{|#1\rangle}

\newtheorem{theorem}{Theorem}

\newtheorem{lemma}{Lemma}

\newtheorem{proposition}{Proposition}

\newtheorem{protocol}{Protocol}

\newtheorem{innercustomthm}{Protocol}
\newenvironment{myprotocol}[1]
  {\renewcommand\theinnercustomthm{#1}\innercustomthm}
  {\endinnercustomthm}

\theoremstyle{definition}
\newtheorem{definition}{Definition}

\definecolor{mypurple2}{RGB}{170,0,255}

\usepackage{soul}
\usepackage{array}

\begin{document} 

\title{Distributing bipartite quantum systems under timing constraints}

\author[a]{Kfir Dolev}
\address{Stanford University}
\ead{dolev@stanford.edu}

\author[b]{Alex May}
\address{The University of British Columbia}
\ead{may@phas.ubc.ca}

\author[c]{Kianna Wan}
\address{Stanford University}
\ead{kianna@stanford.edu}

\begin{abstract}
In many quantum information processing protocols, entangled states shared among parties are an important resource. In this article, we study how bipartite states may be distributed in the context of a quantum network limited by timing constraints. We explore various tasks that plausibly arise in this context, and characterize the achievability of several of these in settings where only one-way communication is allowed. We provide partial results in the case where two-way communication is allowed. This builds on earlier work on summoning single and bipartite systems. 
\end{abstract}

\maketitle
\flushbottom

\tableofcontents
\pagestyle{plain}

\section{Introduction}

Shared bipartite quantum systems, and in particular shared entangled states, are among the most basic resources in quantum information theory. Several developments have demonstrated the usefulness of entangled states in completing tasks that are constrained by the causal structure of spacetime. Two cases of interest are summoning (and its variants)~\cite{kent2012quantum,kent2013no,hayden2016summoning} and position-based cryptography~\cite{kent2011quantum,buhrman2014position}.
Fig.~\ref{fig:singlesystemsummoning} describes an instance of a summoning task, which highlights an interesting use of a maximally entangled state. Note that this entangled state must be distributed between two \emph{spacetime} regions---if the entanglement between the relevant \emph{spatial} regions is established too late, the summoning task cannot be completed. Similarly, entangled states must be distributed between two spacetime regions in order to perform certain position-based cryptographic tasks~\cite{buhrman2014position,beigi2011simplified,tomamichel2013monogamy,dolev2019constraining}.

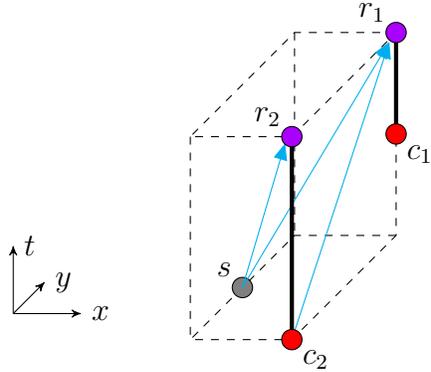
\begin{figure}
\centering
\begin{tikzpicture}[scale=0.9]

	\coordinate (PT) at (0,0,0);
	
	\coordinate (CL) at (1.5,1.5,-2);
	\coordinate (CR) at (1.5,0,2);
	
	\coordinate (QL) at (1.5,3,-2);
	\coordinate (QR) at (1.5,3,2);
	
	\draw[dashed] (1.5,3,2) -- (1.5,3,-2) -- (1.5,0,-2) -- (1.5,0,2);
	\draw[dashed] (1.5,0,-2) -- (0,0,-2) -- (0,0,2) -- (1.5,0,2);
	\draw[dashed] (1.5,3,-2) -- (0,3,-2) -- (0,3,2) -- (1.5,3,2);
	\draw[dashed] (0,3,-2) -- (0,0,-2);
	\draw[dashed] (0,3,2) -- (0,0,2);
	
	\draw[fill=gray] (PT) circle (0.15);
	\node [above left] at (PT) {$s$};
	
	\node [below right] at (CL) {$c_1$};
	\node [below right] at (CR) {$c_2$};
	
	\node [above left] at (QL) {$r_1$};
	\node [above left] at (QR) {$r_2$};
    
    \begin{scope}[shift={($(0,0)$)}]       
    \draw[->] (-3,0,1) -- (-3,0,-0.2);
    \node [right] at (-3,0,-0.2) {$y$};
    \draw[->] (-3,0,1) -- (-2,0,1);
    \node [right] at (-2,0,1) {$x$};
    \draw[->] (-3,0,1) -- (-3,1,1);
    \node [right] at (-3,1,1) {$t$};
    \end{scope}
    
	\draw[cyan][-triangle 45] (PT) -- (1.4,2.86,-2);
	\draw[cyan][-triangle 45] (PT) -- (1.4,2.9,2);
	\draw[cyan][-triangle 45] (CR) -- (1.4,2.86,-2);
    
	\draw[ultra thick] (CL) -- (QL);
	\draw[ultra thick] (CR) -- (QR);
    
	\draw[fill=mypurple2]  (QL) circle (0.15);
	\draw[fill=mypurple2]  (QR) circle (0.15);
    
	\draw[fill=red] (CL) circle (0.15);
	\draw[fill=red] (CR) circle (0.15);
\end{tikzpicture}
\caption{An example of a single-system summoning task [cf.~Definition~\ref{def:singlesystemsummoning}] for which pre-distributed entanglement is necessary. Blue arrows indicate timelike or lightlike curves between spacetime points. At $s$, an unknown quantum state $\ket{\psi}$ is given to Alice by Bob. At each $c_j$, Alice receives a single bit $b_j$, with the promise that $b_{j^*} = 1$ for exactly one $j^* \in \{1,2\}$. Alice's task is to then return $\ket{\psi}$ at $r_{j^{*}}$. This task can be completed with probability $1$, for any choice of $j^* \in \{1,2\}$ by Bob, if entanglement is shared between $s$ and $c_2$ beforehand (see~\cite{hayden2016summoning} for details), but is impossible in the absence of entanglement.\label{fig:singlesystemsummoning}}
\end{figure}

More broadly, there has recently been interest in the development of a distributed quantum network, or ``quantum internet''~\cite{van2014quantum}. In this context, one may be interested in establishing entanglement between distant nodes. Besides the usual restrictions imposed by relativistic causality, the network may have additional timing constraints due to, for instance, the timescales of quantum memories or the properties of communication links. We can then ask between which nodes, and at what times, it is possible to establish entanglement under such constraints.

Motivated by these examples, we investigate how bipartite quantum systems can move through spacetime. We work within the \emph{quantum tasks} framework introduced by Kent~\cite{kent2012quantum}. In this framework, quantum and/or classical systems are received as inputs, processed, and returned as outputs, with the inputs and outputs occurring at various points in spacetime. This framework incorporates summoning and position-based cryptography, as well as other practical tasks~\cite{kent2012unconditionally,kent2011unconditionally, kent1999coin, pitalua2016spacetime,PhysRevA.100.012302, kent2011location}. In addition to these applications, the quantum tasks framework serves as an operational way of recording principles of quantum theory and its interaction with relativity~\cite{kent2012quantum}. It has also been adapted to the holographic setting and used to study the dynamics of quantum information in the context of AdS/CFT \cite{may2019quantum,may2019holographic}.

To characterize the distribution of bipartite quantum systems within the quantum tasks framework, we define several variants on the following task. We consider a set of call-return pairs $\{(c_j,r_j)\}_{j}$, where the $c_j$ and $r_j$ are spacetime points. Each \emph{call point} $c_j$ is in the casual past of its corresponding \emph{return point} $r_j$. An agency\footnote{Throughout this work, we use Alice and Bob to refer to \emph{agencies}. An agency is a collection of collaborating \emph{agents}. The individual agents may move through spacetime and interact with other agents from their own agency or from other agencies.} Alice holds the $A$ and $B$ subsystems of some pure state $\ket{\Psi}_{ABR}$. 
At each call point $c_j$, a second agency Bob will give Alice a classical message $b_j$, with the promise that $b_{j^*} = b_{k^*} = 1$ for exactly two indices $j^*,k^*$. Alice's objective is to return $A$ at the return point $r_j^{*}$ and $B$ at $r_{k^*}$, or vice versa. We say that the task specified by $\{(c_j, r_j)\}_j$ is \emph{achievable} if there exists a strategy that allows Alice to succeed with probability $1$ for \emph{any} choice of $j^{*}, k^{*}$ by Bob. The two basic variants on this task we consider differ in whether $\ket{\Psi}_{ABR}$ is known or unknown to Alice---we refer to the case where $\ket{\Psi}_{ABR}$ is unknown as \emph{two-system summoning}, and the case where it is known as \emph{entanglement summoning}.

At a practical level, the input at a call point $c_j$ may represent a request from a node in a quantum network for some quantum system, and the return point $r_j$ the spacetime location at which that system is needed. The finite time separation between $c_j$ and $r_j$ may be imposed, for instance, by the lifetime of a quantum memory---using $A$ as part of a larger quantum computation may only be possible if some other quantum system is still coherent by the time $A$ is received. 

The most trivial two-system or entanglement summoning task consists of two call-return pairs, as shown in Fig.~\ref{fig:simple}. This task is achievable as long as $s$ is in the causal past of both return points, even if the two call-return pairs are causally disconnected (i.e., $c_1$ cannot send a signal to $r_2$, and vice versa). On the other hand, in the case of three call-return pairs, as in Fig.~\ref{fig:threesimple}, the corresponding task may or may not be achievable. If each call-return pair is causally disconnected from every other pair, then in the case of two-system summoning, Alice is prevented by no-cloning from completing the task (see Lemma~\ref{lemma:noTwoOutGraphs}), while in entanglement summoning, the monogamy of entanglement makes the task unachievable (see Lemma~\ref{lemma:monogamy}). However, if the three call-return pairs are causally connected, the task may be achievable. Thus, these summoning tasks operationally explore how no-cloning or the monogamy of entanglement interacts with causality. 

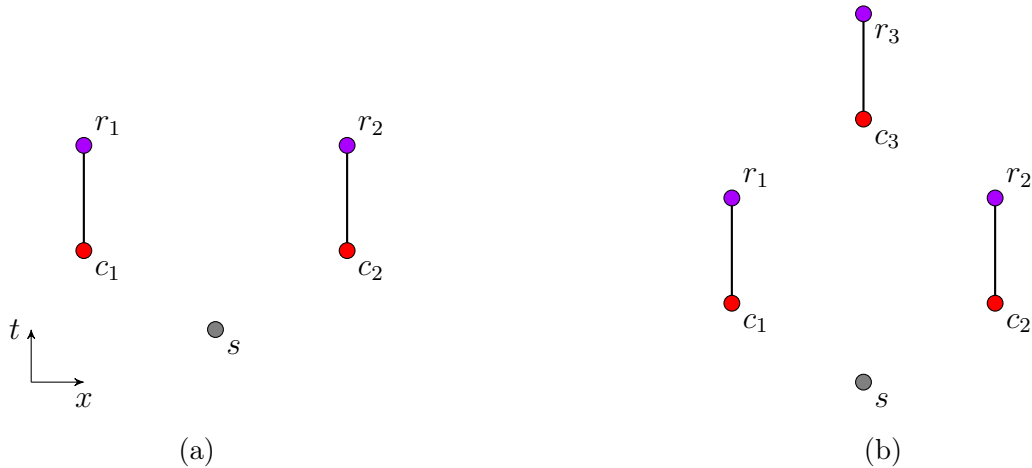
\begin{figure}
\centering
\begin{subfigure}[b]{.45\textwidth}
\centering
\begin{tikzpicture}[scale=0.7]

\coordinate (P) at (0cm, 1cm);

\coordinate (C2) at (2.5cm, 2.5cm);
\coordinate (Q2) at (2.5cm, 4.5cm);
\coordinate (L2) at (3.5cm,3.5cm);
\coordinate (R2) at (1.5cm,3.5cm);

\coordinate (C1) at (-2.5cm,2.5cm);
\coordinate (Q1) at (-2.5cm,4.5cm);
\coordinate (L1) at (-3.5cm,3.5cm);
\coordinate (R1) at (-1.5cm,3.5cm);

\draw[thick] (C1) -- (Q1); 	 
\draw[thick] (C2) -- (Q2); 

\draw[fill=gray] (P) circle (0.15cm);
\node[below right] at (P) {$s$};

\draw[fill=red] (C1) circle (0.15cm);
\node [below right] at (C1) {$c_1$};
\draw[fill=red] (C2) circle (0.15cm);
\node [below right] at (C2) {$c_2$};
	
\draw[fill=mypurple2] (Q1) circle (0.15cm);
\node [above right] at (Q1) {$r_1$};
\draw[fill=mypurple2] (Q2) circle (0.15cm);
\node [above right] at (Q2) {$r_2$};

\draw[->] (-3.5,0) -> (-3.5,1);
\draw[->] (-3.5,0) -> (-2.5,0);

\node [left] at (-3.5,1) {$t$};
\node [below] at (-2.5,0) {$x$};

\node at (0,8) { };

\end{tikzpicture}
\caption{\label{fig:simple}}
\end{subfigure}
\hfill
\begin{subfigure}[b]{.45\textwidth}
\centering
\begin{tikzpicture}[scale=0.7]
  
\coordinate (P) at (0cm, 1cm);

\coordinate (C2) at (2.5cm, 2.5cm);
\coordinate (Q2) at (2.5cm, 4.5cm);
\coordinate (L2) at (3.5cm,3.5cm);
\coordinate (R2) at (1.5cm,3.5cm);

\coordinate (C1) at (-2.5cm,2.5cm);
\coordinate (Q1) at (-2.5cm,4.5cm);
\coordinate (L1) at (-3.5cm,3.5cm);
\coordinate (R1) at (-1.5cm,3.5cm);

\draw[thick] (C1) -- (Q1); 	 
\draw[thick] (C2) -- (Q2); 

\draw[fill=gray] (P) circle (0.15cm);
\node[below right] at (P) {$s$};

\draw[fill=red] (C1) circle (0.15cm);
\node [below right] at (C1) {$c_1$};
\draw[fill=red] (C2) circle (0.15cm);
\node [below right] at (C2) {$c_2$};
	
\draw[fill=mypurple2] (Q1) circle (0.15cm);
\node [above right] at (Q1) {$r_1$};
\draw[fill=mypurple2] (Q2) circle (0.15cm);
\node [above right] at (Q2) {$r_2$};

\draw[thick] (0,6) -- (0,8);
\draw[fill=red] (0,6) circle (0.15cm);
\node [below right] at (0,6) {$c_3$};
\draw[fill=mypurple2] (0,8) circle (0.15cm);
\node [below right] at (0,8) {$r_3$};

\end{tikzpicture}
\caption{\label{fig:threesimple}}
\end{subfigure}
\caption{(a) A two-system summoning task involving two call-return pairs. At $s$, the $A$ and $B$ subsystems of a state $\ket{\Psi}_{ABR}$ are given to Alice. Alice's goal is to bring $A$ to $r_1$ and $B$ to $r_2$, or vice versa. This task is clearly achievable since $r_1$ and $r_2$ are in the casual future of $s$. (b) An achievable two-system summoning task on three call-return pairs. This task can be completed as follows. The $A$ subsystem is sent to $c_1$, and the call information $b_1$ received at $c_1$ is checked. If $b_1=1$, $A$ is brought to $r_1$, whereas if $b_0=0$, $A$ is forwarded to $r_3$. Similarly, the $B$ subsystem is first sent to $c_2$, and then to $r_2$ or $r_3$ based on the value of $b_2$.}
\end{figure}

In~\cite{adlam2018relativistic}, Adlam considered some of the tasks we define here and provided necessary conditions for their completion. We build on this work by constructing explicit protocols for completing a large class of tasks. In particular, when call-return pairs $(c_j,r_j)$, $(c_k,r_k)$ have $c_j\rightarrow r_k$ or $c_k\rightarrow r_j$ but never both, we provide a complete characterization of the tasks defined by Adlam.\footnote{These correspond to the ``unassisted'' variants discussed below.} We do not fully characterize tasks that have $c_j\rightarrow r_k$ and $c_k\rightarrow r_j$, but provide necessary conditions. We also point out that some of Adlam's necessary conditions are incorrect, and provide an explicit counterexample by constructing a protocol for a task that violates these conditions. 

This paper is organized as follows. In Section~\ref{sec:summoningPreliminaries}, we introduce some definitions and notation for describing summoning tasks. In Section~\ref{sec:twosystems}, we consider two-system summoning, which involves responding to requests for a shares of an unknown state. In Section~\ref{sec:entanglementSummoning}, we consider {entanglement summoning}, where the state to be distributed is a known maximally entangled state. Sections \ref{sec:twosystems} and \ref{sec:entanglementSummoning} focus on the case of causal connections that run one way ($c_i\rightarrow r_j$ or $c_j\rightarrow r_i$); in Section \ref{sec:bidirected} we discuss the case where the causal connections may be bidirectional. We conclude in Section~\ref{sec:discussion}.

\section{Preliminaries}\label{sec:summoningPreliminaries}

To motivate the various tasks studied in this paper, consider the following scenario. A number of labs $\{L_j\}_j$ are widely distributed in space. These labs hold quantum computers and may send quantum systems among themselves. In the course of executing some protocol, two labs $L_{j^*}$ and $L_{k^*}$ realize that they need to share the $A$ and $B$ subsystems of some state $\ket{\Psi}_{ABR}$. In the most general setting, $L_{j^*}$ and $L_{k^*}$ may or may not know the state $\ket{\Psi}_{ABR}$, may announce their requests for shares of the state at differing times, and may require their shares before differing deadlines. 

These quantum network scenarios can be abstracted as summoning tasks. We define several variants in this paper to capture a large range of plausible scenarios. In particular, \emph{two-system summoning} considers the case where the requisite quantum state is not known in advance, while \emph{entanglement summoning} considers the case of known resource states. We discuss two-system summoning in Section~\ref{sec:twosystems} and entanglement summoning in Section~\ref{sec:entanglementSummoning}. In this section, we set the stage by introducing terminology and notational conventions that will be used throughout the paper. We also recall the characterization of \emph{single-system summoning}~\cite{kent2013no,hayden2016summoning}, which will be used as a subroutine in some of our two-system and entanglement summoning protocols.

To formalize the notion of labs requesting shares of a quantum state, we introduce call points $c_j$ and return points $r_j$. We associate call and return points into pairs $(c_j,r_j)$, which can be informally thought of as representing a lab $L_j$: the point $c_j$ is the spacetime location at which $L_j$ decides that it needs a share of the state, while $r_j$ is the spacetime location at which the share is needed. Thus, in defining our summoning tasks, we have some (classical) input given at $c_j$ specifying whether a share of the quantum state should be brought to $r_j$. If a share is requested at $r_j$, we will say that a \emph{call} is made to $(c_j, r_j)$, or that $(c_j, r_j)$ is \emph{called} or receives a call.\footnote{We assume that $r_j$ is in the casual future of $c_j$ for every $j$, since a lab would usually decide it needs a state at some point in the future.} In cases where the state is unknown, we must also consider a starting spacetime location $s$ for the state.

The information received at a call point $c_j$ can be forwarded to other spacetime locations. The set of points that can be reached from $c_j$ will be limited, however, either by the finite speed of light or by more practical considerations.\footnote{For instance, in some settings, the relevant speed may instead be the speed of light in a fibre optic cable.} Formally, for two spacetime points $p$ and $q$, we write $p \prec q$ if it is possible to travel from $p$ to $q$. We can then define the \textit{causal future} $J^+(c_j)$ as $J^{+}(c_j) \equiv \{p: c_j \prec p\}$. Similarly, a system may reach the return point $r_j$ from anywhere in its \textit{causal past} $J^-(r_j) \equiv \{p: p \prec r_j\}$ of $r_j$. We refer to the overlap of $J^+(c_j)$ and $J^-(r_j)$ as the \textit{causal diamond} $D_j$. Every point in this region both has access to the classical input at $c_j$ and is able to send information (classical and quantum) to $r_j$.

\begin{definition} \label{def:diamond}
For a pair $(c_j,r_j)$ of call and return points, we define the \emph{causal diamond}
\[ D_j \equiv J^+(c_j) \cap J^-(r_j) \] as the intersection of the causal future of $c_j$ with the causal past of $r_j$.
\end{definition}

Our main technical results on summoning concern the achievability of various classes of summoning tasks. Each summoning task is specified by a set $\{(c_j,r_j)\}_j$ of call-return pairs---equivalently, a set $\{D_j\}_j$ of causal diamonds---and a start point when applicable. To characterize the causal relationships among the diamonds that are relevant to the achievability of summoning tasks, we introduce the following notation.
\begin{definition} \label{def:causallyconnected} For two distinct causal diamonds $D_j$ and $D_k$, we write
\[ D_j \rightarrow D_k \]
iff there exist points $p_j\in D_j$ and $p_k\in D_k$ such that $p_j \prec p_k$. Otherwise, we write $D_j \not\rightarrow D_k$.
\end{definition}
It is straightforward to show that $D_j \rightarrow D_k$ (for $D_j \neq D_k$) iff $c_j\prec r_k$. Consequently, our results could equivalently be phrased solely in terms of call and return points, but the introduction of causal diamonds simplifies notation. If $D_j$ and $D_k$ satisfy both $D_j \to D_k$ and $D_k \to D_j$, we will sometimes write $D_j \leftrightarrow D_k$.

In some cases, the only relevant information is whether at least one of $D_j \to D_k$ and $D_k \to D_j$ is true, so we make the following definition.
\begin{definition}
Two causal diamonds $D_j$ and $D_k$ are said to be \emph{causally connected} if $D_j \rightarrow D_k$ or $D_k \rightarrow D_j$, or both. We write $D_j \sim D_k$ in this case. If $D_j$ and $D_k$ are not causally connected, we write $D_j \not\sim D_k$.
\end{definition}

In earlier work, the scenario in which a call is made to just one of the diamonds, requesting a single share of an unknown quantum state, has been fully characterized~\cite{hayden2016summoning}. We refer to this as \emph{single-system summoning}, and define it below. This definition is adapted from \cite{hayden2019localizing}.
\begin{definition}\label{def:singlesystemsummoning}
A \emph{single-system summoning task} is specified by a start point $s$ and a set of causal diamonds $\{D_j\}_{j=1}^n$, each of which is defined by a call point $c_j$ and return point $r_j$ (with $c_j \prec r_j$). The task involves two agencies, Alice and Bob, which perform the following steps.
\begin{enumerate}[1.]
    \item Bob prepares a quantum state $\ket{\Psi}_{AR}$, and gives Alice the $A$ subsystem at $s$.
    \item At each $c_j$, Bob gives Alice a classical bit $b_j$. Alice is promised that $b_{j^*} = 1$ for exactly one $j^* \in\{1,\dots, n\}$, but does not know the value of $j^*$ in advance.
\end{enumerate}
The task is \emph{achievable} iff for any state $\ket{\Psi}_{AR}$ and any choice of $j^* \in \{1,\dots, n\}$ by Bob, Alice can return subsystem $A$ at $r_{j^*}$ with success probability $1$.
\end{definition}
Single-system summoning is neatly characterized by the following theorem, proven in~\cite{hayden2016summoning}. 
\begin{theorem} \label{thm:single}
A single-system summoning task is achievable if and only if
\begin{enumerate*}
    \item the return point $r_j$ of every diamond $D_j$ is in the causal future of the start point $s$, and
    \item every pair of diamonds  $D_j $ and $D_k$ is causally connected.
\end{enumerate*}
\end{theorem}
We will use single-system summoning as a subroutine in some of our protocols for two-system and entanglement summoning. Explicit protocols for completing single-system summoning tasks are given in~\cite{hayden2016summoning,hayden2016spacetime,wu2018efficient}. In \cite{hayden2016summoning,hayden2019localizing}, it was also argued that single-system summoning can be interpreted in terms of the localization of quantum information to sets of spacetime regions. It is natural to ask if a similar conclusion holds in the context of two-system summoning. We discuss this in \ref{sec:localizingentanglement}.

In the following sections, we will generally assume that $s$ is in the distant past of all of the diamonds---more precisely, that $s \prec c_j$ for all call points $c_j$. However, any protocol that completes a summoning task for $s$ in the distant past can be easily modified into a protocol that requires only $s \prec r_j$ for all return points $r_j$ (as in Condition~(i) of Theorem~\ref{thm:single}). We describe how to do this in the case of single-system summoning; the same strategy applies to the two-system summoning tasks we consider later. Suppose that the input system $A$ is located at a point $s$ such that $s \prec r_j$ for all $r_j$. Prepare a Bell state $\ket{\Phi}_{EE'}$ at some point $p$ in the distant past, and execute the original summoning protocol (which works for systems starting in the distant past) with input $E$ and start point $p$. Bring $E'$ to $s$, measure $AE'$ in the Bell basis, and broadcast the measurement outcome into the causal future of $s$. At the return point of the called diamond, the original summoning protocol produces $A$, up to a Pauli correction. Since every $r_j$ is in the causal future of $s$, it is then possible to apply the appropriate Pauli correction and return $A$. 

Note that provided that $s$ is in the causal past of all $r_j$, the achievability of a single-system summoning task depends only on the diamonds being causally connected---the directionality of the connections is irrelevant. 
In contrast, for the two-system summoning tasks discussed in the subsequent sections, the directionality of the causal connections between diamonds is often important. Hence, given the set of causal diamonds specifying a summoning task, it will be useful to summarize the relationships among the diamonds using a directed graph, which we define below.
\begin{definition} \label{def:causalgraph}
For a set of causal diamonds $\mathcal{D} = \{D_1, \dots, D_n\}$, the corresponding \textit{causal graph} $G_{\mathcal{D}} = (V,E)$ has vertex set $V = \mathcal{D}$, with $(D_j, D_k) \in E$ iff $D_j \to D_k$.
\end{definition}
If $D_j \leftrightarrow D_k$, then both $(D_j, D_k)$ and $(D_k,D_j)$ are edges in $G_{\mathcal{D}}$, in which case we will collectively refer to $(D_j,D_k)$ and $(D_k,D_j)$ as a \textit{bidirected edge}. An \textit{oriented graph} is a directed graph with no bidirected edges, i.e., at most one of $(D_j, D_k)$ and $(D_k, D_j)$ may be an edge. 

In the following, we will simply refer to the causal graph corresponding to the set of diamonds specifying a summoning task as the causal graph corresponding to that task. Since different sets of diamonds may give rise the same causal graph, different tasks may have the same corresponding causal graph.

As an example, Condition~(ii) of Theorem~\ref{thm:single} can be reformulated in terms of a task's corresponding causal graph as follows. For every pair of vertices $D_j$ and $D_k$, at least one of $(D_j, D_k)$ and $(D_k, D_j)$ is an edge in the causal graph. In other words, the causal graph contains a \emph{tournament} (an oriented graph in which every pair of vertices shares an edge) as a spanning subgraph.\footnote{A spanning subgraph of a graph $G$ is a subgraph of $G$ that includes all of the vertices of $G$.}

Fig.~\ref{fig:spacetimeandCausalgraph} shows an instance of a task embedded in spacetime and its corresponding causal graph. Note that the causal graph does not capture all of the information about the placement of call and return points. In particular, for two different arrangements of call and return points that have the same causal graph, there may exist an intermediate point $g$ such that $c_i,c_j \prec g \prec r_k,r_l$ for one arrangement, but not for the other. However, none of the protocols or necessary conditions presented in this article rely on any causal features beyond those captured by the causal graph. This is consistent with no-go theorems for position-based cryptography \cite{kent2011quantum,buhrman2014position} and with a more general result in \cite{dolev2019constraining}. 

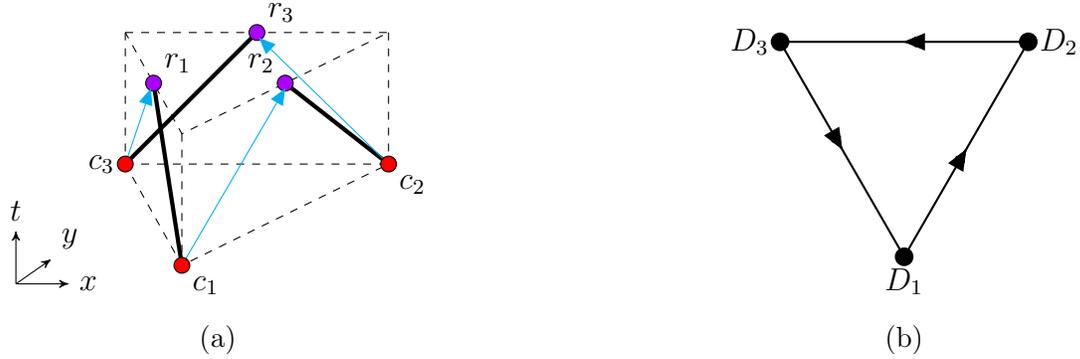
\begin{figure}
\centering
\begin{subfigure}[b]{.45\textwidth}
\centering
\begin{tikzpicture}[scale=0.7]

\coordinate (PS) at (2.5,-3.5,1.5);

	\coordinate (Ca) at (0,0,0);
	\coordinate (Cb) at (5,0,0);
	\coordinate (Cc) at (3,0,5);
	
	\coordinate (Qa) at (2.5,2.5,0);
	\coordinate (Qb) at (4, 2.5, 2.5);
	\coordinate (Qc) at (1.5, 2.5, 2.5);
	
	\draw[dashed] (Ca) -- (Cb) -- (Cc) -- (Ca);
	\draw[dashed] (0,2.5,0) -- (5,2.5,0) -- (3,2.5,5) -- (0,2.5,0);
	\draw[dashed] (0,2.5,0) -- (0,0,0);
	\draw[dashed] (5,2.5,0) -- (5,0,0);
	\draw[dashed] (3,2.5,5) -- (3,0,5);
	
	\draw[ultra thick] (Ca) -- (Qa);
	\draw[ultra thick] (Cb) -- (Qb);
	\draw[ultra thick] (Cc) -- (Qc);
	
	\draw[fill=mypurple2] (Qa) circle (0.15);
	\draw[fill=mypurple2]  (Qb) circle (0.15);
	\draw[fill=mypurple2]  (Qc) circle (0.15);
	
	\node [above right] at (Qa) {$r_3$};
	\node [above left] at (Qb) {$r_2$};
	\node [above right] at (Qc) {$r_1$};
	
	\draw[-triangle 45][cyan] (Ca) -- (1.45, 2.4, 2.5);
	\draw[-triangle 45][cyan] (Cb) -- (2.54,2.4,0);
	\draw[-triangle 45][cyan] (Cc) -- (4, 2.4, 2.5);
	
	\draw[fill=red] (Ca) circle (0.15);
	\draw[fill=red] (Cb) circle (0.15);
	\draw[fill=red] (Cc) circle (0.15);
	
	\node [left] at (Ca) {$c_3$};
	\node [below right] at (Cb) {$c_2$};
	\node [below right] at (Cc) {$c_1$};
	
	\draw[->] (-1.5,-1.7,1.5) -> (-0.5,-1.7,1.5);
	\draw[->] (-1.5,-1.7,1.5) -> (-1.3,-1.7,0.3);
	\draw[->] (-1.5,-1.7,1.5) -> (-1.5,-0.7,1.5);
	
	\node[right] at (-0.5,-1.7,1.5) {$x$};
	\node[above right] at (-1.3,-1.7,0.3) {$y$};
	\node[above] at (-1.5,-0.7,1.5) {$t$};

\end{tikzpicture}
\caption{}
\end{subfigure}
\hfill
\begin{subfigure}[b]{.45\textwidth}
\centering
\begin{tikzpicture}[scale=0.55]
\begin{scope}[thick,decoration={
    markings,
    mark=at position 0.5 with {\arrow{triangle 45}}}
    ] 
  
\draw[fill=black] (0,0) circle (0.2cm);
\node[below] at (0,0) {$D_1$};

\draw[postaction={on each segment={mid arrow}}] (-3,5.2) -- (0,0);

\draw[fill=black] (-3,5.2) circle (0.2cm);
\node[left] at (-3,5.2) {$D_3$};

\draw[postaction={on each segment={mid arrow}}] (3,5.2) -- (-3,5.2);

\draw[fill=black] (3,5.2) circle (0.2cm);
\node[right] at (3,5.2) {$D_2$};

\draw[postaction={on each segment={mid arrow}}] (0,0) -- (3,5.2);

\end{scope}
\end{tikzpicture}
\caption{\label{fig:threecyclegraph}}
\end{subfigure}
\caption{An example of a set of three call-return pairs, or equivalently, three causal diamonds. (a) The arrangement of the call and return points in spacetime. In this case, each diamond is a single line (shown in black), since for each $j \in \{1,2,3\}$, $c_j$ and $r_j$ are lightlike separated. Blue arrows between $p$ and $q$ indicate that $p \prec q$. (b) The corresponding causal graph [cf.~Definition~\ref{def:causalgraph}], which we refer to as a ``three-cycle.''\label{fig:spacetimeandCausalgraph}}
\end{figure}

\section{Two-system summoning}\label{sec:twosystems}

In this section, we define and characterize two-system summoning tasks, which involve distributing the $A$ and $B$ subsystems of an unknown quantum state $\ket{\Psi}_{ABR}$ (to the return points of) two causal diamonds. In \S\ref{sec:unassisted}, we consider the case where the calls are \emph{unassisted}, meaning that a call to $c_j$ reveals only whether a share of the state is required at $r_j$ [cf.~Definition~\ref{def:UBS}]. This seems the most likely scenario in the context of a quantum network. Two further subsections deal with variants on this task that involve \emph{assistance}, provided in the form of extra information at the call points. The first variant simply breaks the symmetry of the calls [cf.~Definition~\ref{def:LBS}], while in the second variant, each called diamond is notified of the identity of both called diamonds [cf.~Definition~\ref{def:GBS}]. These also correspond to plausible network scenarios, and allow us to explore the subtleties of how limiting the information at the call points affects the achievability of summoning tasks.

\subsection{No assistance} \label{sec:unassisted}

We define unassisted two-system summoning following~\cite{adlam2018relativistic}.\footnote{See Definition~1 therein. Note that in~\cite{adlam2018relativistic}, our ``two-system summoning'' is referred to as ``entanglement summoning.''}
\begin{definition} \label{def:UBS}
An \emph{unassisted two-system summoning task} is specified by a start point $s$ and a set of causal diamonds $\{D_j\}_{j=1}^n$, each of which is defined by a call point $c_j$ and return point $r_j$. The task involves two agencies, Alice and Bob, which perform the following steps.
\begin{enumerate}[1.]
    \item Bob prepares a quantum state $\ket{\Psi}_{ABR}$, and gives Alice the $A$ and $B$ subsystems at $s$.
    \item At each $c_j$, Bob gives Alice a classical bit $b_j$. Alice is promised that $b_{j^*} = b_{k^*} = 1$ for exactly two $j^*, k^* \in \{1,\dots, n\}$, but does not know the values of $j^*$ and $k^*$ in advance.
\end{enumerate}
The task is \emph{achievable} iff for any state $\ket{\Psi}_{ABR}$ and any choice of $j^*, k^* \in \{1,\dots, n\}$ by Bob, Alice can return subsystem $A$ at $r_{j^*}$ and subsystem $B$ at $r_{k^*}$, or vice versa, with success probability~$1$.
\end{definition}

In this subsection, we characterize the achievability of all unassisted two-system summoning tasks with no bidirectional causal connections between diamonds. The causal graphs corresponding to such tasks are all oriented graphs.  
Note, however, that if a given oriented graph corresponds to achievable tasks, then any graph obtained from this graph by making some of the edges bidirected also corresponds to achievable tasks. One can simply run the protocol for the oriented graph, ignoring the extra edges. Causal graphs with bidirected edges are explored further in Section~\ref{sec:bidirected}.

To analyze these tasks, we begin by studying examples involving only three causal diamonds. We simplify by moving the start point $s$ to the distant past. As explained after Theorem \ref{thm:single}, any summoning task which is achievable with $s$ in the distant past is also achievable when $s$ is merely in the causal past of all of the $r_i$. 

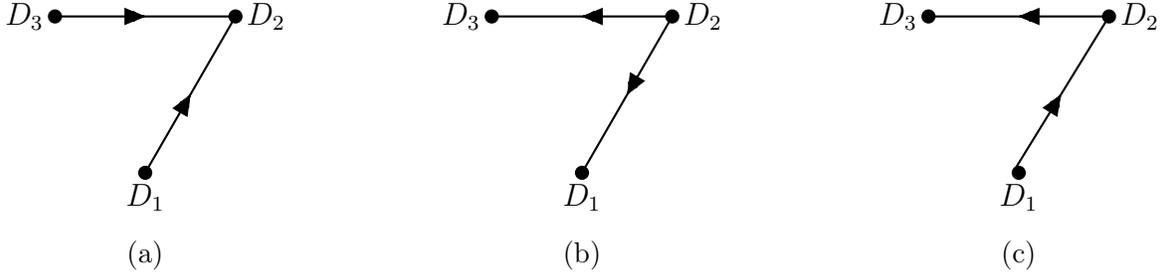
\begin{figure}
\centering
\begin{subfigure}[b]{.3\textwidth}
\centering
\begin{tikzpicture}[scale=0.4]
\begin{scope}[thick,decoration={
    markings,
    mark=at position 0.5 with {\arrow{triangle 45}}}
    ] 
  
\draw[fill=black] (0,0) circle (0.2cm);
\node[below] at (0,0) {$D_1$};

\draw[fill=black] (-3,5.2) circle (0.2cm);
\node[left] at (-3,5.2) {$D_3$};

\draw[postaction={on each segment={mid arrow}}] (-3,5.2) -- (3,5.2);

\draw[fill=black] (3,5.2) circle (0.2cm);
\node[right] at (3,5.2) {$D_2$};

\draw[postaction={on each segment={mid arrow}}] (0,0) -- (3,5.2);

\end{scope}

\end{tikzpicture}
\caption{\label{fig:twoIn}}
\end{subfigure}
\hfill
\begin{subfigure}[b]{.3\textwidth}
\centering
\begin{tikzpicture}[scale=0.4]
\begin{scope}[thick,decoration={
    markings,
    mark=at position 0.5 with {\arrow{triangle 45}}}
    ] 
  
\draw[fill=black] (0,0) circle (0.2cm);
\node[below] at (0,0) {$D_1$};

\draw[fill=black] (-3,5.2) circle (0.2cm);
\node[left] at (-3,5.2) {$D_3$};

\draw[postaction={on each segment={mid arrow}}] (3,5.2) -- (-3,5.2);

\draw[fill=black] (3,5.2) circle (0.2cm);
\node[right] at (3,5.2) {$D_2$};

\draw[postaction={on each segment={mid arrow}}] (3,5.2) -- (0,0);

\end{scope}
\end{tikzpicture}
\caption{\label{fig:twoOut}}
\end{subfigure}
\hfill 
\begin{subfigure}[b]{.3\textwidth}
\centering
\begin{tikzpicture}[scale=0.4]
\begin{scope}[thick,decoration={
    markings,
    mark=at position 0.5 with {\arrow{triangle 45}}}
    ] 
  
\draw[fill=black] (0,0) circle (0.2cm);
\node[below] at (0,0) {$D_1$};

\draw[fill=black] (-3,5.2) circle (0.2cm);
\node[left] at (-3,5.2) {$D_3$};

\draw[postaction={on each segment={mid arrow}}] (3,5.2) -- (-3,5.2);

\draw[fill=black] (3,5.2) circle (0.2cm);
\node[right] at (3,5.2) {$D_2$};

\draw[postaction={on each segment={mid arrow}}] (-0.2,0) -- (3,5.2);

\end{scope}
\end{tikzpicture}
\caption{\label{fig:inandout}}
\end{subfigure}

\caption{(a) The ``two-in'' graph. This is the causal graph corresponding to the task shown in Fig.~\ref{fig:threesimple}. (b) The ``two-out'' graph. (c) The ``in-and-out'' graph.\label{fig:twoInAndTwoOut}}
\end{figure}

Consider the ``two-in'' causal graph shown in Fig.~\ref{fig:twoIn}. Any unassisted two-system summoning task with this graph is achievable. To complete it, bring the $A$ subsystem to $D_1$ (from $s$) and the $B$ subsystem to $D_3$. Then, at $D_1$, return $A$ at $r_1$ if $b_1 = 1$, whereas if $b_1=0$, forward $A$ to $D_2$. The protocol at $D_3$ is analogous. One can readily verify that calls to any two of $D_1,D_2,D_3$ result in the $A$ and $B$ subsystems being returned at (the return points of) the two called diamonds. Note that it immediately follows that any task on three diamonds whose causal graph contains this ``two-in'' graph as a subgraph is also achievable. 

Now, consider the ``two-out'' causal graph in Fig.~\ref{fig:twoOut}. Our first lemma establishes the unachievability of all unassisted two-system summoning tasks with this causal graph.
\begin{lemma}\label{lemma:noTwoOutGraphs}
Any unassisted two-system summoning task corresponding to the ``two-out'' causal graph in Fig.~\ref{fig:twoOut}, or any spanning subgraph thereof, is unachievable.
\end{lemma}
\begin{proof}
Suppose that there exists a protocol that completes such a task. This means that when calls are made to any pair of diamonds, the protocol returns subsystems $A$ and $B$ at the return points of the two called diamonds. Consider the case where $D_2$ is one of the called diamonds. Since $D_1,D_3\not\rightarrow D_2$, the subsystem returned at (the return point of) $D_2$ cannot depend on what the other called diamond is, so we can take this subsystem to be $A$ without loss of generality. Then, making the other call to $D_1$ should result in $B$ being returned to $D_1$, whereas making the other call being to $D_3$ should result in $B$ being returned to $D_3$. 

Now, suppose that Bob cheats by making calls to all three diamonds. Since $D_1 \not\sim D_3$, i.e., $D_1$ and $D_3$ are not causally connected, the behaviour of the protocol at $D_1$ cannot be affected by the call to $D_3$, and vice versa. Consequently, the fact that $D_1$ and $D_3$ separately return $B$ when, respectively, $(D_1,D_2)$ and $(D_3,D_2)$ are called implies that they will both return $B$ when all three diamonds are called. But this implies that $B$ has been cloned, in violation of the linearity of quantum mechanics. Even though calling three diamonds breaks the rules of the task, it should not be possible to manipulate any physically allowed protocol to violate a law of physics. Hence, no such protocol exists, so the task is unachievable. 
\end{proof}

To characterize all achievable tasks on three diamonds, it remains to consider the graph in Fig~\ref{fig:threecyclegraph}. This is a directed cycle graph of length $3$, which we will refer to as a ``3-cycle.'' Unassisted two-system summoning tasks with this causal graph are also unachievable.
\begin{lemma}\label{lemma:nothreecycle}
Any unassisted two-system summoning task corresponding to the ``3-cycle'' causal graph shown in Fig.~\ref{fig:threecyclegraph}, or any spanning subgraph thereof, is unachievable. 
\end{lemma}
\begin{proof}
Suppose that there exists a protocol that completes such a task. This means that: 
\begin{enumerate}
    \item Calls to $D_1$ and $D_2$ result in $A$ and $B$ each returned at $D_1$ and $D_2$ (i.e., $A$ at $D_1$ and $B$ at $D_2$, or vice versa). 
    \item Calls to $D_2$ and $D_3$ result in $A$ and $B$ returned at $D_2$ and $D_3$.
    \item Calls to $D_3$ and $D_1$ result in $A$ and $B$ returned at $D_3$ and $D_1$.
\end{enumerate}
Now, suppose that Bob cheats by making calls to $D_1$, $D_2$, and $D_3$. Then, consider the information available at (the return point of) $D_1$. $D_1$ only sees the calls to $D_1$ and $D_3$, and therefore has the same information as in the case where only $D_1$ and $D_3$ are called. Hence, by~(iii), $D_1$ returns $A$ or $B$. Similarly, $D_2$ only sees the calls to $D_2$ and $D_1$, so by~(i), $D_2$ must return $A$ or $B$. Finally, $D_3$ only sees the calls to $D_2$ and $D_3$, so by (ii), $D_3$ returns $A$ or $B$. But this implies that the protocol must have cloned either $A$ or $B$, which is impossible, so no such protocol exists. 
\end{proof}

We can now characterize tasks involving an arbitrary number of diamonds. Note that any vertex-induced subgraph of the full causal graph corresponds to an unachievable task, then the full task is also unachievable. Hence, Lemmas~\ref{lemma:noTwoOutGraphs} and~\ref{lemma:nothreecycle} can be used to constrain the set of achievable tasks. Conversely, we prove in~\S\ref{sec:UBSprotocol} that a task is achievable as long as none of the vertex-induced subgraphs of its causal graph coincide with the unachievable three-vertex graphs of Lemmas~\ref{lemma:noTwoOutGraphs} and~\ref{lemma:nothreecycle}.
Thus, we arrive at necessary and sufficient conditions for the achievability of unassisted two-system summoning tasks.

\begin{theorem}[unassisted two-system summoning] \label{thm:unassistedtwosystemsummoning} An unassisted two-system summoning task whose corresponding causal graph is an oriented graph is achievable if and only if
\begin{enumerate}
    \item the return point $r_j$ of every diamond $D_j$ is in the causal future of the start point $s$, and
    \item in the subgraph of the causal graph induced by any three vertices, there are two edges pointing into one of the vertices.
\end{enumerate}
\end{theorem}
\begin{proof}
The necessity of Condition (i) is immediate from causality. Now, suppose that a subgraph $G_3$ of the full causal graph $G_{\mathcal{D}}$ induced by some set of three vertices violates Condition~(ii), i.e., each vertex in $G_3$ has in-degree at most~$1$. Clearly, $G_3$ must be a spanning subgraph of either the ``two-out'' graph (Fig.~\ref{fig:twoOut}) or the ``3-cycle'' (Fig.~\ref{fig:threecyclegraph}). In either case, any unassisted two-system summoning task corresponding $G_3$ is unachievable, by Lemmas~\ref{lemma:noTwoOutGraphs} and~\ref{lemma:nothreecycle}. Since $G_3$ is a vertex-induced subgraph of $G_{\mathcal{D}}$, it immediately follows that any unassisted two-system summoning task corresponding to $G_{\mathcal{D}}$ is also unachievable. 

The proof that Conditions~(i) and~(ii) are sufficient proceeds in two steps. First, we show in Lemma~\ref{lemma:unassistedgraphs} in~\ref{appendix:graphlemmas} that any oriented graph satisfying Condition~(ii) falls under one of the two cases in Fig.~\ref{fig:twosystemgraphs}: the graph is either a transitive tournament (Fig.~\ref{fig:transitive}), or obtained from a transitive tournament by removing the edge between the first two vertices in its topological ordering (Fig.~\ref{fig:nearlyTransitive}). We then provide protocols for both cases in \S\ref{sec:UBSprotocol}.
\end{proof}

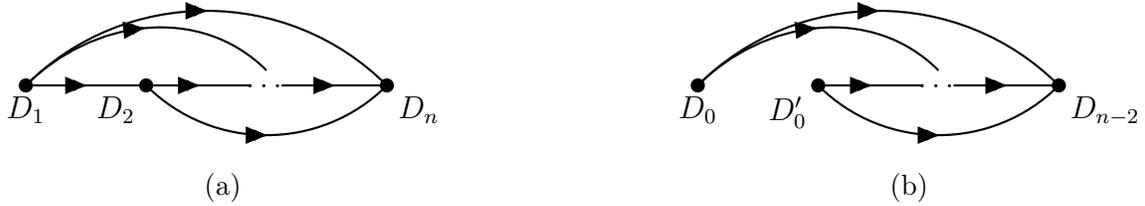
\begin{figure}
\centering
\begin{subfigure}[b]{.45\textwidth}
\centering
\begin{tikzpicture}[scale=0.4]
\begin{scope}[thick,decoration={
    markings,
    mark=at position 0.5 with {\arrow{triangle 45}}}
    ] 
  
\draw[fill=black] (0,0) circle (0.2cm);
\node[below] at (0,0) {$D_1$};

\draw[postaction={on each segment={mid arrow}}] (0,0) -- (4,0);

\draw[fill=black] (4,0) circle (0.2cm);
\node[below left] at (4,0) {$D_2$};

\draw[postaction={on each segment={mid arrow}}] (4,0) -- (7.5,0);

\draw[fill=black] (12,0) circle (0.2cm);
\node[below right] at (12,0) {$D_n$};

\node at (8,0) {\dots};

\draw[postaction={on each segment={mid arrow}}] (8.5,0) -- (12,0);

\draw[postaction={on each segment={mid arrow}}] (0,0) to [out=45,in=135] (8,0.5);
\draw[postaction={on each segment={mid arrow}}] (0,0) to [out=45,in=135] (12,0);

\draw[postaction={on each segment={mid arrow}}] (4,0) to [out=-45,in=-135] (12,0);

\end{scope}

\end{tikzpicture}
\caption{\label{fig:transitive}}
\end{subfigure}
\hfill
\begin{subfigure}[b]{.45\textwidth}
\centering
\begin{tikzpicture}[scale=0.4]

\begin{scope}[thick,decoration={
    markings,
    mark=at position 0.5 with {\arrow{triangle 45}}}
    ] 
  
\draw[fill=black] (0,0) circle (0.2cm);
\node[below] at (0,0) {$D_0$};


\draw[fill=black] (4,0) circle (0.2cm);
\node[below left] at (4,0) {$D_0'$};

\draw[postaction={on each segment={mid arrow}}] (4,0) -- (7.5,0);

\draw[fill=black] (12,0) circle (0.2cm);
\node[below right] at (12,0) {$D_{n-2}$};

\node at (8,0) {\dots};

\draw[postaction={on each segment={mid arrow}}] (8.5,0) -- (12,0);

\draw[postaction={on each segment={mid arrow}}] (0,0) to [out=45,in=135] (8,0.5);
\draw[postaction={on each segment={mid arrow}}] (0,0) to [out=45,in=135] (12,0);

\draw[postaction={on each segment={mid arrow}}] (4,0) to [out=-45,in=-135] (12,0);

\end{scope}
\end{tikzpicture}\caption{\label{fig:nearlyTransitive}}
\end{subfigure}
\caption{Schematic diagrams illustrating graphs corresponding to achievable unasssisted two-system summoning tasks. (a) A \textit{transitive tournament}: a tournament that has a \textit{topological ordering}---the vertices can be labelled such that $j < k$ if $D_j\rightarrow D_k$. (b) A graph that nearly forms a transitive tournament: the graph is the same as in (a), except with the edge between the first two vertices in the topological ordering removed. We then relabel the first two diamonds as $D_0$ and $D_0'$.\label{fig:twosystemgraphs}}
\end{figure}

\subsubsection{Protocol for unassisted two-system summoning} \label{sec:UBSprotocol}~\\

Consider an unassisted two-system summoning task that satisfies the conditions of Theorem~\ref{thm:unassistedtwosystemsummoning}. By Lemma~\ref{lemma:unassistedgraphs} in~\ref{appendix:graphlemmas}, this implies that the causal graph has one of the two forms shown in Fig.~\ref{fig:twosystemgraphs}. We start by constructing a protocol (Protocol~\ref{protocol:entanglementRailsA}) for the simpler case of Fig.~\ref{fig:transitive}, where the graph is a transitive tournament. Then, we extend this protocol to the case of Fig.~\ref{fig:nearlyTransitive} (Protocol~\ref{protocol:entanglementRailsB}). For simplicity, the protocols are described for the case where $A$ and $B$ are qubit systems; the extension to larger systems is trivial.

In the following protocols, the diamonds are labelled as in Fig.~\ref{fig:twosystemgraphs}. In particular, note that for each $i \geq 1$, $D_i \to D_j$ for all $j > i$. This means that any measurement outcomes obtained at the call point $c_i$ can be broadcast to all return points $r_j$ with $j > i$.

\pagebreak
\begin{myprotocol}{1a}\label{protocol:entanglementRailsA} ~\\
\noindent \ul{Preparation:}
\begin{itemize}
    \item For each $i \in \{1,\dots, n-1\}$, distribute Bell pairs $\ket{\Phi}_{E_i'E_{i+1}}$ and $\ket{\Phi}_{F_i'F_{i+1}}$ between the call points of $D_i$ and $D_{i+1}$ [cf.~Fig.~\ref{fig:entanglementRails}]. 
    \item Send subsystems $A$ and $B$ (from the start point $s$) to the call point of diamond $D_1$. Relabel $A$ as $E_1$ and $B$ as $F_1$.
\end{itemize}
\ul{Execution:}
\begin{itemize}
    \item At $D_i$ for each $i \in \{1,\dots, n\}$:
    \begin{itemize}
    \item If a call is received at the call point $c_i$ (i.e., $b_i = 1$), return $E_i$ at the return point $r_i$.\footnote{Implicitly, any extra Pauli operators on the state due to the teleportation are to be removed at $r_j$ using the teleportation data (measurement outcomes) received from diamonds in the causal past of $D_i$.\label{Paulicorrectionfootnote}} Measure $F_iE_{i}'$ in the Bell basis and broadcast the measurement outcome to every $r_j$ with $j>i$.
    \item If no call is received at $c_i$ (i.e., $b_i = 0$), measure $E_iE_{i}'$ in the Bell basis, and measure $F_iF_i'$ in the Bell basis. Broadcast the measurement outcomes to every $r_j$ with $j>i$.
    \end{itemize}
\end{itemize}
\end{myprotocol}

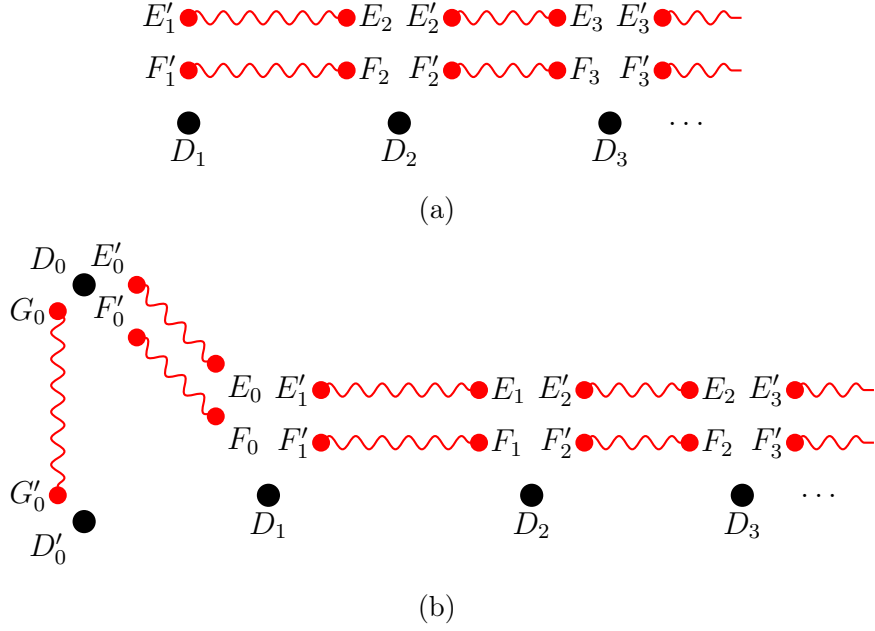
\begin{figure}
\centering
\begin{subfigure}[a]{1\textwidth} 
\centering
\begin{tikzpicture}[scale=0.7]
\begin{scope}[thick,decoration={
    markings,
    mark=at position 0.5 with {\arrow{triangle 45}}}
    ] 
 
\draw[fill=black] (0,0) circle (0.2cm);
\node[below] at (0,-0.1) {$D_1$};

\draw[fill=black] (4,0) circle (0.2cm);
\node[below] at (4,-0.1) {$D_2$};

\draw[fill=black] (8,0) circle (0.2cm);
\node[below] at (8,-0.1) {$D_3$};

\draw[red,fill=red] (0,1) circle (0.15cm);
\node[left] at (0,1) {$F_1'$};
\draw[red,fill=red] (3,1) circle (0.15cm);
\node[right] at (3,1) {$F_2$};
\draw[red,snake it] (0,1) -- (3,1);

\draw[red,fill=red] (0,2) circle (0.15cm);
\node[left] at (0,2) {$E_1'$};
\draw[red,fill=red] (3,2) circle (0.15cm);
\node[right] at (3,2) {$E_2$};
\draw[red,snake it] (0,2) -- (3,2);

\draw[red,fill=red] (5,1) circle (0.15cm);
\node[left] at (5,1) {$F_2'$};
\draw[red,fill=red] (7,1) circle (0.15cm);
\node[right] at (7,1) {$F_3$};
\draw[red,snake it] (5,1) -- (7,1);

\draw[red,fill=red] (5,2) circle (0.15cm);
\node[left] at (5,2) {$E_2'$};
\draw[red,fill=red] (7,2) circle (0.15cm);
\node[right] at (7,2) {$E_3$};
\draw[red,snake it] (5,2) -- (7,2);

\draw[red,fill=red] (9,1) circle (0.15cm);
\node[left] at (9,1) {$F_3'$};
\draw[red,snake it] (9,1) -- (10.5,1);

\draw[red,fill=red] (9,2) circle (0.15cm);
\node[left] at (9,2) {$E_3'$};
\draw[red,snake it] (9,2) -- (10.5,2);

\node at (9.5,0) {\dots};

\end{scope}

\end{tikzpicture}
\caption{\label{fig:entanglementRails}}
\end{subfigure} 

\begin{subfigure}[b]{1\linewidth}
\centering
\begin{tikzpicture}[scale=0.7]
\begin{scope}[thick,decoration={
    markings,
    mark=at position 0.5 with {\arrow{triangle 45}}}
    ] 
 
\draw[fill=black] (-1,0) circle (0.2cm);
\node[below] at (-1,-0.1) {$D_1$};

\draw[fill=black] (4,0) circle (0.2cm);
\node[below] at (4,-0.1) {$D_2$};

\draw[fill=black] (8,0) circle (0.2cm);
\node[below] at (8,-0.1) {$D_3$};

\draw[fill=black] (-4.5,4) circle (0.2cm);
\node[above left] at (-4.6,4) {$D_0$};

\draw[fill=black] (-4.5,-0.5) circle (0.2cm);
\node[below left] at (-4.6,-0.5) {$D_0'$};

\draw[red,fill=red] (0,1) circle (0.15cm);
\node[left] at (0,1) {$F_1'$};
\draw[red,fill=red] (3,1) circle (0.15cm);
\node[right] at (3,1) {$F_1$};
\draw[red,snake it] (0,1) -- (3,1);

\draw[red,fill=red] (0,2) circle (0.15cm);
\node[left] at (0,2) {$E_1'$};
\draw[red,fill=red] (3,2) circle (0.15cm);
\node[right] at (3,2) {$E_1$};
\draw[red,snake it] (0,2) -- (3,2);

\draw[red,fill=red] (5,1) circle (0.15cm);
\node[left] at (5,1) {$F_2'$};
\draw[red,fill=red] (7,1) circle (0.15cm);
\node[right] at (7,1) {$F_2$};
\draw[red,snake it] (5,1) -- (7,1);

\draw[red,fill=red] (5,2) circle (0.15cm);
\node[left] at (5,2) {$E_2'$};
\draw[red,fill=red] (7,2) circle (0.15cm);
\node[right] at (7,2) {$E_2$};
\draw[red,snake it] (5,2) -- (7,2);

\draw[red,fill=red] (9,1) circle (0.15cm);
\node[left] at (9,1) {$F_3'$};
\draw[red,snake it] (9,1) -- (10.5,1);

\draw[red,fill=red] (9,2) circle (0.15cm);
\node[left] at (9,2) {$E_3'$};
\draw[red,snake it] (9,2) -- (10.5,2);

\node at (9.5,0) {\dots};

\draw[red,snake it] (-5,3.5) -- (-5,0);
\draw[red,fill=red] (-5,3.5) circle (0.15cm);
\node[left] at (-5,3.5) {$G_0$};
\draw[red,fill=red] (-5,0) circle (0.15cm);
\node[left] at (-5,0) {$G_0'$};

\draw[red,snake it] (-3.5,3) -- (-2,1.5);
\draw[red,fill=red] (-3.5,3) circle (0.15cm);
\node[above left] at (-3.5,3) {$F_0'$};
\draw[red,fill=red] (-2,1.5) circle (0.15cm);
\node[below right] at (-2,1.5) {$F_0$};

\draw[red,snake it] (-3.5,4) -- (-2,2.5);
\node[above left] at (-3.5,4) {$E_0'$};
\draw[red,fill=red] (-3.5,4) circle (0.15cm);
\draw[red,fill=red] (-2,2.5) circle (0.15cm);
\node[below right] at (-2,2.5) {$E_0$};
\end{scope}
\end{tikzpicture}
\caption{\label{fig:entanglementRails2}
}
\end{subfigure}
\caption{(a) The distribution of Bell pairs (red) used in Protocol~\ref{protocol:entanglementRailsA}, which completes unassisted two-system summoning in the case where the causal graph has the form in Fig.~\ref{fig:transitive}. (b) The distribution of Bell pairs (red) used in Protocol~\ref{protocol:entanglementRailsB}, which completes unassisted two-system summoning in the case where the causal graph has the form in Fig.~\ref{fig:nearlyTransitive}.}
\end{figure}

The set-up of the pre-distributed entanglement is shown in Fig.~\ref{fig:entanglementRails}. The $E_i'E_{i+1}$ Bell pairs collectively form an ``entanglement rail'', and the $F_i' F_{i+1}$ Bell pairs form another rail. At the diamonds that are not called, the $A$ subsystem is forwarded along the $E$ rail via teleporation, and the $B$ subsystem along the $F$ rail. Let the two called diamonds be $D_{j^*}$ and $D_{k^*}$, with $j^* < k^*$. Since none of the diamonds $D_i$ with $i< j^*$ are called, $A$ is teleported (without the Pauli corrections performed) to $E_{j^*}'$. $D_{j^*}$ receives all of the measurement outcomes from these teleportations at $r_{j^*}$, and is therefore able to return $A$ at $r_{j^*}$. According to the protocol, $D_{j^*}$ also teleports $B$ from $F_{j^*}'$ to $E_{j^* + 1}$, which effectively moves $B$ from the $F$ rail to the $E$ rail. Thus, $D_{k^*}$ returns $B$ when it hands in $E_{k^*}$ as per the protocol.

Next, we build on this protocol to handle the case in Fig.~\ref{fig:nearlyTransitive}.

\begin{myprotocol}{1b}\label{protocol:entanglementRailsB}~\\
\noindent \ul{Preparation:}
\begin{itemize}
    \item For each $i \in \{0,\dots, n-3\}$, distribute Bell pairs $\ket{\Phi}_{E_i'E_{i+1}}$ and $\ket{\Phi}_{F_i'F_{i+1}}$ between the call points of $D_i$ and $D_{i+1}$. Additionally, distribute a Bell pair $\ket{\Phi}_{G_0 G_0'}$ between the call points of $D_0$ and $D_0'$ [cf.~Fig.~\ref{fig:entanglementRails}].
    \item Send subsystem $A$ (from the start point $s$) to the call point of $D_0$, and subsystem $B$ to the call point of $D_0'$.
\end{itemize}
\ul{Execution:}
\begin{itemize}
    \item At $D_0$:
    \begin{itemize}
        \item If a call is received at the call point $c_0$, return $A$ at the return point $r_0$. Measure $G_0 E_0'$ in the Bell basis and broadcast the measurement outcome to every $r_j$ with $j \geq 1$.
        \item If no call is received at $c_0$, measure $AE_0'$ in the Bell basis, and measure $G_0 F_0'$ in the Bell basis. Broadcast the measurement outcomes to every $r_j$ with $j \geq 1$.
    \end{itemize}
    \item At $D_0'$:
    \begin{itemize}
        \item If a call is received at the call point $c_0'$, return $B$ at return point $r_0'$. 
        \item If no call is received at $c_0'$, measure $BG_0'$ in the Bell basis and broadcast the measurement outcome to every $r_j$ with $j \geq 1$. 
    \end{itemize}
    \item At $D_i$ for each $i \in \{1,\dots, n-2\}$, perform the same steps as in Protocol~\ref{protocol:entanglementRailsA}.
\end{itemize}
\end{myprotocol}

The arrangement of the pre-distributed entanglement is shown in Fig.~\ref{fig:entanglementRails2}. Similar to in Protocol~\ref{protocol:entanglementRailsA}, there are two entanglement rails, $E$ and $F$, along diamonds $D_0, \dots, D_{n-2}$. In addition, a Bell pair is shared between the only pair of disconnected diamonds, $D_0$ and $D_0'$.
To understand why this protocol works, consider the possible combinations of called diamonds. In the case where both $D_0$ and $D_0'$ are called, the $A$ and $B$ subsystems are returned at $D_0$ and $D_0'$ respectively, completing the task. If neither $D_0$ nor $D_0'$ are called, $D_0$ teleports $A$ to $E_0$ and $D_0'$ teleports $B$ to $F_0$, and the remaining diamonds---which form a transitive tournament---execute Protocol~\ref{protocol:entanglementRailsA}, teleporting $A$ and $B$ along the two rails until they are handed in by the called diamonds. If $D_0$ is called and $D_0'$ is not, $D_0$ returns $A$, and the $B G_0'$ and $G_0E_0'$ measurements by $D_0'$ and $D_0$ result in $B$ (modulo Pauli corrections) being moved to the $E$ rail. The steps executed at the remaining diamonds teleport $B$ along the $E$ rail until it is handed in by the other called diamond. Finally, if $D_0'$ is called and $D_0$ is not, $D_0'$ returns $B$, and $A$ is teleported along the $E$ rail to the other called diamond. 

\subsection{Assistance through labelled calls} \label{sec:LABS}

In a two-system summoning task, the goal is to distribute an unknown quantum between two spacetime locations. An interesting example is the ``$3$-cycle'' graph shown in Fig.~\ref{fig:threecyclegraph}. By Lemma~\ref{lemma:nothreecycle}, any unassisted two-system summoning task corresponding to this graph is unachievable. In some ways, the unachievability of this task is surprising, since any \textit{single}-system summoning task with the same graph is achievable~\cite{hayden2016summoning}, and adding in a second system and a second call appears to make the task easier, as either system can go to either called diamond. The fact that more options actually makes the task more difficult represents a ``paradox of choice'', similar to the one observed in~\cite{adlam2016quantum}.

It turns out that to make tasks corresponding to the ``$3$-cycle'' graph achievable, it suffices to break the symmetry of the two calls. Bob can do this by giving Alice an extra resource: an anti-correlated pair of bits. Equivalently, he labels the two calls $1$ and $2$ (instead of giving the number $1$ to both called diamonds, as in the unassisted setting). In this section, we explore the consequences of providing this extra information.

Labelling the two calls defines a new summoning task, which we formalize below. 
\begin{definition} \label{def:LBS} A \emph{label-assisted two-system summoning task} is specified by a start point $s$ and a set of causal diamonds $\{D_j\}_{j=1}^n$, each of which is defined by a call point $c_j$ and return point $r_j$. The task involves by two agencies, Alice and Bob, which perform the following steps:
\begin{enumerate}
    \item Bob prepares a quantum state $\ket{\Psi}_{ABR}$, and gives Alice the $A$ and $B$ subsystems at $s$.
    \item At each $c_j$, Bob gives Alice a number $b_j\in\{0,1,2\}$. Alice is promised that $b_{j^*} = 1$ for exactly one $j^*\in \{1,\dots, n\}$ that and $b_{k^*} = 2$ for exactly one $k^*\in \{1,\dots,n\}$, but does not know the values of $j^*$ and $k^*$ in advance.
\end{enumerate}
The task is \emph{achievable} if for any state $\ket{\Psi}_{ABR}$ and any choice of $j^*, k^* \in \{1,\dots, n\}$ by Bob, Alice can return subsystem $A$ at $r_{j^*}$ and subsystem $B$ at $r_{k^*}$, or vice versa, with success probability~$1$.
\end{definition}

It is straightforward to complete any label-assisted two-system summoning task corresponding to the ``$3$-cycle.'' To do this, simply execute two single-system summoning protocols in parallel: $A$ is used as the input for a single-system summoning task that responds to the call labelled $1$, while $B$ is used as the input for a separate single-system summoning task that responds to the call labelled $2$.

The reader may wonder why this strategy does not work in the case of unassisted two-system summoning. To clarify this, let us discuss in more detail how Alice takes advantage of the labelled calls to perform two single-system summoning tasks. In general, to perform one such task, Alice will have an agent positioned at the call point of each diamond. The agent performs actions which depend on whether their diamond is called. Now, if Bob makes two calls, one labelled $1$ and the other labelled $2$, Alice can summon a subsystem to each by deploying two separate, non-interacting groups of agents, with the first group ``ignoring'' the $1$ call (i.e., treating the $1$ call as no call), and the second group ignoring the $2$ call. Clearly, this is only possible when the two calls have different labels. If we the unassisted case, if we set up two single-system summoning tasks, the agent at the call point of a called diamond would not know which of these tasks it should ``ignore,'' as it cannot coordinate with the agent at the other called diamond.

Although this type of assistance renders tasks with the ``$3$-cycle'' graph achievable, tasks with the ``two-out'' graph in Fig.~\ref{fig:twoOut} is still unachievable in the label-assisted setting.
\begin{lemma}\label{lemma:LabelAssisted-NoTwoOut}
Any label-assisted two-system summoning task corresponding to the ``two-out'' causal graph in Fig.~\ref{fig:twoOut}, or any spanning subgraph thereof, is unachievable.
\end{lemma}
\begin{proof}The same argument as is used in the proof of Lemma~\ref{lemma:noTwoOutGraphs} applies---consider making a call labelled $1$ to $D_2$, and calls labelled $2$ to $D_1$ and $D_3$. 
\end{proof}

To complete the characterization of tasks on three diamonds, we also need to consider the ``in-and-out'' graph in Fig.~\ref{fig:inandout}.
\begin{lemma}\label{lemma:no_a->b->c}
Any label-assisted two-system summoning task corresponding to the ``in-and-out'' causal graph in Fig.~\ref{fig:inandout}, or any spanning subgraph thereof, is unachievable.
\end{lemma} 

\begin{proof}
Suppose that there exists a protocol that completes such a task. This means that: 
\begin{enumerate}
    \item If Bob makes a call labelled $1$ to $D_1$ and a call labelled $2$ to $D_2$, $A$ and $B$ are returned at $D_1$ and $D_2$ (i.e., $A$ is returned at $D_1$ and $B$ at $D_2$, or vice versa.
    \item If Bob makes a call labelled $2$ to $D_2$ and a call labelled $1$ to $D_3$, $A$ and $B$ are returned at $D_2$ and $D_3$.
\end{enumerate}
Now, suppose that Bob cheats by making \emph{two} calls labelled $1$ to $D_1$ and $D_3$, and a call labelled $2$ to $D_2$. Since $D_3 \not\rightarrow D_1$, what is returned at $D_1$ cannot depend on whether $D_3$ is called, so by~(i), $D_1$ returns $A$ or $B$. Similarly, since $D_3 \not\rightarrow D_2$, by~(i), $D_2$ returns $A$ or $B$. Finally, since $D_1 \not\rightarrow D_3$, $D_3$ acts as it would in the case where Bob only makes a call labelled $2$ to $D_2$ and a call labelled $1$ to $D_3$, so by (ii), $D_3$ returns $A$ or $B$. But this implies that either $A$ or $B$ has been cloned, which is impossible, so no such protocol exists.
\end{proof}

Combining Lemmas \ref{lemma:LabelAssisted-NoTwoOut} and \ref{lemma:no_a->b->c}, we see that in the causal graph of any achievable task (on an arbitrary number of diamonds), every subgraph induced by three vertices must either have two edges pointing into one vertex, or form a ``$3$-cycle.'' In fact, this condition is also sufficient for the task to be achievable, as we prove in \ref{appendix:labelassistedtwosystem}. 
\begin{theorem}[label-assisted two-system summoning] \label{thm:labelassistedtwosystemsummoning} A label-assisted two-system summoning task whose causal graph is an oriented graph is achievable if and only if
\begin{enumerate}
    \item the return point $r_j$ of every diamond $D_j$ is in the causal future of the start point $s$, and 
    \item the subgraph of the causal graph induced by any three vertices either has two edges pointing into one of the vertices, or is a cycle of length $3$.
\end{enumerate}
\end{theorem}
\begin{proof}
The necessity of Condition~(i) is immediate from causality. The necessity of Condition~(ii) follows immediately from Lemmas~\ref{lemma:LabelAssisted-NoTwoOut} and~\ref{lemma:no_a->b->c}, along with the fact if any vertex-induced subgraph of the causal graph corresponds to an unachievable task, then the full task is also unachievable.

We prove that Conditions~(i) and~(ii) are sufficient in \ref{appendix:labelassistedtwosystem}. First, we show using Lemma~\ref{lemma:onemissingedge} in \ref{appendix:graphlemmas} that any oriented graph satisfying Condition~(ii) falls under one of the two cases in Fig.~\ref{fig:rivalGraph}. We then construct protocols for both cases.
\end{proof}

\begin{figure}
\centering
\begin{subfigure}[b]{.45\textwidth}
\centering
\begin{tikzpicture}[scale=0.65]
    \begin{scope}[thick,decoration={
    markings,
    mark=at position 0.5 with {\arrow{triangle 45}}}
    ] 
    
    \draw[fill=black] (0,0) circle (0.1cm);
    \node[below] at (0,0) {$D_1$};
    
    \draw[fill=black] (2,2) circle (0.1cm);
    \node[right] at (2,2) {$D_3$};
    
    \draw[fill=black] (-2,2) circle (0.1cm);
    \node[left] at (-2,2) {$D_2$};

    \node[below] at (0,4) {\dots};
    
      
    \node[below] at (-3,-3) {$ $};
    
    \node[below] at (3,-3) {$ $};

    
    \draw (-2,2) -- (-1,3);
    \draw (2,2) -- (1,3);
    \draw (0,0) -- (0,3);
    
    \draw[postaction={on each segment={mid arrow}}] (2,2) -- (0,0);
    \draw[postaction={on each segment={mid arrow}}] (-2,2) -- (0,0);
    \draw[postaction={on each segment={mid arrow}}] (-2,2) -- (2,2);

    \end{scope}
    \end{tikzpicture}
\caption{\label{fig:tournament}}
\end{subfigure}
\hfill
\begin{subfigure}[b]{.45\textwidth}
\centering
\begin{tikzpicture}[scale=0.65]
    \begin{scope}[thick,decoration={
    markings,
    mark=at position 0.5 with {\arrow{triangle 45}}}
    ] 
    
    \draw[fill=black] (0,0) circle (0.1cm);
    \node[below] at (0,0) {$D_1$};
    
    \draw[fill=black] (2,2) circle (0.1cm);
    \node[right] at (2,2) {$D_3$};
    
    \draw[fill=black] (-2,2) circle (0.1cm);
    \node[left] at (-2,2) {$D_2$};

    \node[below] at (0,4) {\dots};
    
    \draw[dashed,blue] (0,2) circle (3.1);
      
    \draw[fill=black] (-3,-3) circle (0.1cm);
    \node[below] at (-3,-3) {$D_0$};
    
    \draw[fill=black] (3,-3) circle (0.1cm);
    \node[below] at (3,-3) {$D_0'$};

    \draw[-triangle 45] (-3,-3) -- (-1.5,-0.75);
    \draw[-triangle 45] (3,-3) -- (1.5,-0.75);
    
    \draw (-2,2) -- (-1,3);
    \draw (2,2) -- (1,3);
    \draw (0,0) -- (0,3);
    
    \draw[postaction={on each segment={mid arrow}}] (2,2) -- (0,0);
    \draw[postaction={on each segment={mid arrow}}] (-2,2) -- (0,0);
    \draw[postaction={on each segment={mid arrow}}] (-2,2) -- (2,2);

    \end{scope}
    \end{tikzpicture}
\caption{\label{fig:rivalGraphb}}
\end{subfigure}
\caption{Schematic diagrams illustrating graphs corresponding to achievable label-assisted two-system summoning tasks. (a)~A tournament: every pair of vertices shares an edge. (b)~A graph in which there are only two vertices, $D_0$ and $D_0'$, that do not share an edge. The remaining vertices $\{D_j\}_{j=1}^{n-2}$ induce a tournament, and $D_0,D_0' \to D_j$ for every $j \geq 1$. 
\label{fig:rivalGraph}}
\end{figure}
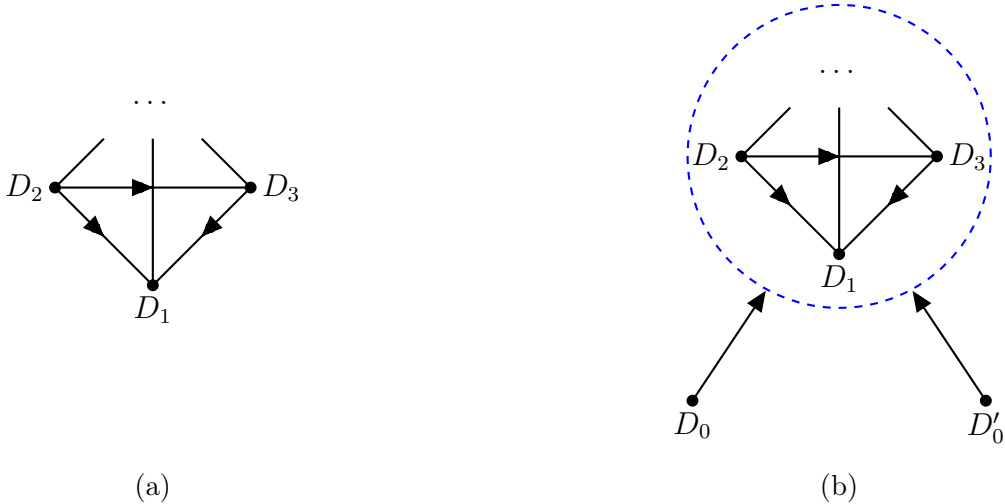

Thus, substantially more tasks are achievable in the label-assisted case than in the unassisted case. Recall that in the unassisted case, the causal graph needs to have a topological ordering, and either every pair of vertices shares an edge, as in Fig.~\ref{fig:transitive}, or there is one pair of non-adjacent vertices, as in Fig.~\ref{fig:nearlyTransitive}. With labelled calls, the need for topological ordering is relaxed, and the graph takes the form of Fig.~\ref{fig:tournament} if there are no non-adjacent vertices, or Fig.~\ref{fig:rivalGraphb} when there is one pair of non-adjacent vertices. This parallels results on single-system summoning. Consider single-system summoning as specified in Definition~\ref{def:singlesystemsummoning}, but allow Bob to give multiple calls. If the calls are unlabelled, \cite{adlam2016quantum} showed that only tournaments with a topological ordering (Fig.~\ref{fig:transitive}) correspond to achievable tasks, due to the ``paradox of choice'' mentioned earlier. However, if the calls are labelled, Alice can ignore all labels except one, reducing the many-call task to the usual single-call task. Thus, the achievability of labelled many-call tasks is also characterized by Theorem~\ref{thm:single}, so the causal graph needs only be a tournament (Fig.~\ref{fig:tournament}). 

\subsection{Assistance through global calls}\label{sec:twoSystemGlobalCalls}

Bob could also provide additional information to Alice in the form of ``global'' calls. Recall that in the unassisted case, if a diamond is called, it does not receive any information (at its call point) about the identity of the other called diamond. In the global-assisted case, on the other hand, each of the two called diamonds $D_{j^*}$ and $D_{k^*}$ is given both $j^*$ and $k^*$. 

\begin{definition} \label{def:GBS} A \emph{global-assisted two-system summoning task} is specified by a start point $s$ and a set of causal diamonds $\{D_j\}_{j=1}^n$, each of which is defined by a call point $c_j$ and return point $r_j$. The task involves two agencies, Alice and Bob, which perform the following steps:
\begin{enumerate}[1.]
    \item Bob prepares a quantum state $\ket{\Psi}_{ABR}$ and gives Alice the $A$ and $B$ subsystems at $s$.
    \item At each $c_j$, Bob gives Alice either $0$, or the tuple $(j^*,k^*)$ with $j^*,k^* \in \{1,\dots,n\}$ and $j^* \neq k^*$. Alice is promised that she will receive $(j^*,k^*)$ at $c_{j^*}$ and $c_{k^*}$, and $0$ at all of the other call points.
\end{enumerate}
The task is \emph{achievable} if for any state $\ket{\Psi}_{ABR}$ and any choice of $j^*,k^* \in \{1,\dots, n\}$ by Bob, Alice can return subsystem $A$ at $r_{j^*}$ and subsystem $B$ at $r_{k^*}$, or vice versa, with success probability~$1$.
\end{definition}

We do not fully characterize global-assisted two-system summoning, but we make one observation. In the global-assisted case, there are achievable tasks whose causal graphs have two or more pairs of non-adjacent vertices. This stands in contrast to the unassisted and label-assisted cases, where there can be at most one pair of non-adjacent vertices, [cf.~Figs.~\ref{fig:twosystemgraphs} and~\ref{fig:rivalGraph}], i.e., at most one pair of diamonds that is not causally connected. As an example, consider the causal graph in Fig.~\ref{fig:twomissingpossible} (which clearly violates the conditions of Theorems~\ref{thm:unassistedtwosystemsummoning} and~\ref{thm:labelassistedtwosystemsummoning}). There are two pairs of non-adjacent vertices ($D_1 \not\sim D_3$ and $D_2\not\sim D_4$), but nonetheless, any global-assisted two-system summoning task with this graph can be completed. We give an explicit protocol in \ref{appendix:globalassistedexample}.

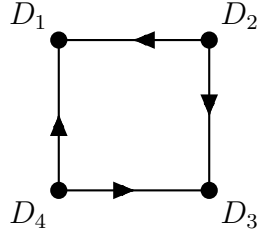
\begin{figure}
    \centering
    \begin{tikzpicture}[scale=1]
    \begin{scope}[thick,decoration={
    markings,
    mark=at position 0.5 with {\arrow{triangle 45}}}
    ] 
  
    \draw[fill=black] (-1,1) circle (0.1cm);
    \node[above left] at (-1,1) {$D_1$};

    \draw[fill=black] (1,1) circle (0.1cm);
    \node[above right] at (1,) {$D_2$};

    \draw[fill=black] (1,-1) circle (0.1cm);
    \node[below right] at (1,-1) {$D_3$};
    
    \draw[fill=black] (-1,-1) circle (0.1cm);
    \node[below left] at (-1,-1) {$D_4$};

    \draw[postaction={on each segment={mid arrow}}] (1,1) -- (-1,1);
    \draw[postaction={on each segment={mid arrow}}] (1,1) -- (1,-1);
    \draw[postaction={on each segment={mid arrow}}] (-1,-1) -- (1,-1);
    \draw[postaction={on each segment={mid arrow}}] (-1,-1) -- (-1,1);

    \end{scope}
    \end{tikzpicture}
    \caption{A four-vertex graph with two pairs of non-adjacent vertices. Any global-assisted two-system summoning task corresponding to this graph is achievable (see \ref{appendix:globalassistedexample}), but its unassisted and label-assisted counterparts are not.\label{fig:twomissingpossible}}
\end{figure}

\section{Entanglement summoning}\label{sec:entanglementSummoning}

In the two-system summoning tasks discussed in the previous section, the goal is to fulfill requests for two subsystems of an \emph{unknown} quantum state to be brought to two parties (of many) in different spacetime regions, abstracted as causal diamonds. However, in the context of many practical applications, the parties may instead need to share a \textit{known} resource state, for use in a later protocol. In this setting, it is more natural to consider entanglement summoning, where the goal is to distribute the $A$ and $B$ subsystems of a Bell pair $\ket{\Phi}_{AB}$.

Since the state to be distributed is known, many copies may be prepared in advance. Consequently, the no-cloning theorem does not restrict the achievability of entanglement summoning tasks, as it did for two-system summoning. Instead, we find restrictions that originate from the monogamy of entanglement, which can create a coordination issue. In particular, one might try to complete an entanglement summoning task by simply distributing many copies $\{\ket{\Phi}_{A_i B_i}\}_i$ between all of the pairs of causal diamonds in advance. However, this would result in each diamond having multiple halves of Bell pairs, and since each subsystem $A_i$ is maximally entangled with only one other subsystem $B_i$, the two called diamonds must somehow ensure that they return halves of the \textit{same} Bell pair. The goal of this section is to understand for which causal graphs this coordination problem can be solved. 

As in Section~\ref{sec:twosystems}, we restrict our attention to tasks with one-way causal connections between diamonds in this section. We discuss tasks with bidirectional causal connections in Section~\ref{sec:bidirected}.

\subsection{No assistance}

We begin with the unassisted case, where the only information given to Alice at each call point $c_j$ is whether a share of a Bell pair should be returned to the corresponding return point $r_j$. Our definition follows \cite{adlam2018relativistic}.\footnote{See Definition~2 therein. Note that in~\cite{adlam2018relativistic}, our ``entanglement summoning'' is referred to as ``entanglement distribution.''}
\begin{definition} \label{def:UES}
An \textit{unassisted entanglement summoning task} is specified by a set of causal diamonds $\{D_j\}_{j=1}^n$, each of which is defined by a call point $c_j$ and return point $r_j$. The task involves two agencies, Alice and Bob, which perform the following:
\begin{enumerate}[1.]
    \item At each $c_j$, Bob gives Alice a classical bit $b_j\in\{0,1\}$. Alice is promised that $b_{j^*} = b_{k^*} = 1$ for exactly two $j^*,k^* \in \{1,\dots, n\}$, but does not know the values of $j^*$ and $k^*$ in advance.
\end{enumerate}
The task is \emph{achievable} if for any choice of $j^*,k^* \in \{1,\dots,n\}$ by Bob, Alice can return a system $A$ at $r_{j^*}$ and a system $B$ at $r_{k^*}$ such that the state on $AB$ is the Bell pair $ \ket{\Phi}\equiv \frac{1}{\sqrt{2}}(\ket{00}+\ket{11})$.\footnote{The task is defined with respect to a Bell pair, but our conclusions are unchanged for any bipartite state $\rho$ (on any number of qubits) that is \emph{not 2-shareable} \cite{werner1990remarks,ranade2009symmetric,myhr2009spectrum}. A state $\rho$ is $2$-shareable if there exists a density matrix $\sigma_{ABB'}$ such that the density matrix on $AB$ and $AB'$ are both equal to $\rho$. This suffices for the proof of Lemma~\ref{lemma:monogamy} below, and the protocol in Theorem~\ref{thm:entanglementSummoningUnassisted}
can be adapted to distribute the state $\rho$.}
\end{definition}
As in the previous section, to determine what causal graphs correspond to achievable tasks, we start by analyzing tasks on three diamonds.
\begin{lemma}\label{lemma:monogamy}
Any unassisted entanglement summoning task corresponding to the ``two-out'' causal graph in Fig.~\ref{fig:twoOut}, or any spanning subgraph thereof, is unachievable. 
\end{lemma}
\begin{proof}
Suppose that there exists a protocol that completes such a task. This means that in the case where $D_2$ and $D_1$ are called, $D_2$ returns a system $A$ and $D_1$ returns a system $B_1$, where the state on $AB_1$ is $\ket{\Phi}$. Now, consider the case where $D_2$ and $D_3$ are called. Since $D_1, D_3 \not\rightarrow D_2$, the system that $D_2$ returns cannot be influenced by which of $D_1$ or $D_3$ is called. Therefore, $D_2$ must return $A$, as it does in the case where only $D_2$ and $D_1$ are called. Meanwhile, $D_3$ returns a system $B_3$ such that the state on $AB_3$ is $\ket{\Phi}$. 

Suppose that Bob cheats by making calls to all three diamonds. Since $D_1, D_3 \not\rightarrow D_2$, we know that $D_2$ still returns $A$. Also, $D_1 \not\sim D_3$, so $D_1$ returns $B_1$ such that the state on $AB_1$ is $\ket{\Phi}$, as it would in the case where only $D_2$ and $D_1$ are called. By the same argument, $D_3$ returns $B_3$ such that the state on $AB_3$ is $\ket{\Phi}$. But by the monogamy of entanglement, $B_1$ and $B_3$ cannot both be maximally entangled with $A$. Therefore, no such protocol exists.
\end{proof}

Lemma~\ref{lemma:monogamy} helps us prove Theorem~\ref{thm:entanglementSummoningUnassisted} below, which provides necessary and sufficient conditions for the achievability of unassisted entanglement summoning tasks on any number of diamonds.
To understand the intuition behind these conditions, consider the information that is available to one of the two called diamonds, say $D_{j^*}$. In a generic protocol, several Bell pairs $\ket{\Phi}_{A_iB_i}$ will have been distributed beforehand, so $D_{j^*}$ may have shares of multiple Bell pairs. Of these shares, $D_{j^*}$ must return the one that is entangled with the share returned by the other called diamond, $D_{k^*}$. In the case where $D_{k^*} \rightarrow D_{j^*}$, Alice's agent at (the return point of) $D_{j^*}$ would see that the other call was made to $D_{k^*}$, and would therefore return their share of the Bell pair shared with $D_{k^*}$. However, if $D_{k^*} \not\rightarrow D_{j^*}$, Alice's agent at $D_{j^*}$ potentially encounters a coordination issue. In this case, the agent sees that none of the diamonds $D_i$ for which $D_i \to D_{j^*}$ have been called, so they know that the other called diamond is in the set
$\mathcal{S}_{j^*} \equiv \{D_i : D_i \not\rightarrow D_{j^*}\}$, 
but does not know which one it is. Hence, it is natural to expect that in any successful protocol, there is a system $A_{j^*}$ such that whenever one of the diamonds in $\mathcal{S}_{j^*}$ is called, that diamond returns $A_{j^*}$. The agent at $D_{j^*}$ can then confidently return the system $B_{j^*}$ for which the state on $A_{j^*}B_{j^*}$ is $\ket{\Phi}$
. Since any pair of diamonds may be called, this argument applies to every $j^* \in \{1,\dots, n\}$. 

The above reasoning motivates a necessary condition: for each $j \in \{1,\dots, n\}$, it should be possible to complete a single-system summoning task on the diamonds in $\mathcal{S}_j$. By Theorem~\ref{thm:single}, this requires each pair of diamonds in $\mathcal{S}_j$ to be causally connected. Equivalently, the subgraph of the full causal graph induced by $\mathcal{S}_j$ must be a tournament. 
We prove that this condition is in fact necessary and sufficient, establishing the following theorem.
\begin{theorem}[unassisted entanglement summoning]\label{thm:entanglementSummoningUnassisted} 
An unassisted entanglement summoning task whose causal graph $G_{\mathcal{D}} = (\mathcal{D},E)$ is an oriented graph is achievable if and only if for every $D_j \in \mathcal{D}$, the subgraph $G_{\mathcal{S}_j}$ of $G_{\mathcal{D}}$ induced by
\begin{equation} \label{Sj}
    \mathcal{S}_j \equiv \{D_i \in \mathcal{D}: D_i \not\rightarrow D_j\}
\end{equation}
is a tournament.
\end{theorem}
\begin{proof}
We prove necessity by contradiction. Suppose that for some $k \in \{1,\dots, n\}$, $G_{\mathcal{S}_k}$ is not a tournament, i.e., there exist diamonds $D_l, D_m \in \mathcal{S}_k$ that are not causally connected, $D_l \not\sim D_m$. By the definition of $\mathcal{S}_k$, we also have $D_l \not\rightarrow D_k$ and $D_m \not\rightarrow D_k$. Hence, the subgraph of $G_{\mathcal{D}}$ induced by $\{D_k, D_l, D_m\}$ must be a subgraph of the ``two-out'' graph with $D_k \to D_l$ and $D_k \to D_m$. By Lemma~\ref{lemma:monogamy}, any unassisted entanglement summoning task corresponding to such a graph is unachievable. It follows that any task corresponding to the full causal graph $G_{\mathcal{D}}$ is unachievable.

We prove sufficiency below in \S~\ref{sec:UESprotocol}, by constructing a protocol for completing any unassisted entanglement summoning task satisfying the conditions of the theorem. 
\end{proof}

Of course, any task with bidirectional causal connections whose causal graph can be obtained by adding edges to an oriented graph satisfying the conditions of Theorem~\ref{thm:entanglementSummoningUnassisted} is also achievable.

\subsubsection{Protocol for unassisted entanglement summoning} \label{sec:UESprotocol}~\\

Consider an unassisted entanglement summoning task, specified by a set of causal diamonds $\mathcal{D} = \{D_j\}_{j=1}^n$, that satisfies the conditions of Theorem~\ref{thm:entanglementSummoningUnassisted}. Namely, the corresponding causal graph $G_{\mathcal{D}}$ is an oriented graph, and the induced subgraph $G_{\mathcal{S}_j}$ for every $j \in \{1,\dots, n\}$ is a tournament. We provide a protocol (Protocol~\ref{UESprotocol}) that completes any such task.

First, note that the conditions of Theorem~\ref{thm:entanglementSummoningUnassisted} imply that for each $D_j \in \mathcal{D}$, there is at most one $D_{j'} \in \mathcal{D}$ for which $D_j \not\sim D_{j'}$, i.e., every diamond is only causally disconnected from at most one other diamond. To see this, suppose to the contrary that $D_j \not\sim D_{j'}$ and $D_j \not\sim D_{j''}$ for $j' \neq j''$. Then, by Eq.~\eqref{Sj}, $D_{j'}, D_{j''} \in \mathcal{S}_j$, and since $\mathcal{S}_j$ must be a tournament, we must have $D_{j'} \sim D_{j''}$.  Since $G_{\mathcal{D}}$ is an oriented graph, either $D_{j'} \to D_{j''}$ or $D_{j''} \to D_{j'}$, but not both. In the former case, $D_j, D_{j''} \in \mathcal{S}_{j'}$, and $D_j \not\sim D_{j''}$ implies that $\mathcal{S}_{j'}$ is not a tournament, while in the latter case, $D_j, D_{j'} \in \mathcal{S}_{j''}$, and $D_j \not\sim D_{j'}$ implies that $\mathcal{S}_{j''}$ is not a tournament.

Hence, defining the subset
\begin{equation} \mathcal{X}  \equiv \{D_j \in \mathcal{D}: \exists D_i \in \mathcal{D} \text{ for which } D_j \not\sim D_i\}, \end{equation}
we know that for every $D_j \in \mathcal{X}$, $D_j \not\sim D_{j'}$ for exactly one $D_{j'} \in \mathcal{X}$. This means that $\mathcal{X}$ can be partitioned into pairs of diamonds $\{D_j, D_{j'}\}$, where $D_j \not\sim D_{j'}$, but $D_j \sim D_k$ and $D_{j'} \sim D_k$ for all $k \not\in \{j,j'\}$. We denote this partition by $\mathcal{X}_{\textrm{part}}$.

In the protocol below, we will use single-system summoning as a black-box subroutine.\footnote{One can use one of the explicit single-system summoning protocols described in \cite{hayden2016summoning,hayden2016spacetime,wu2018efficient}.} Recall from Theorem~\ref{thm:single} that a single-system summoning task on a set of causal diamonds $\mathcal{S}$ is achievable if the causal graph $G_{\mathcal{S}}$ corresponding to $\mathcal{S}$ is a tournament. Note that as it is defined in Definition~\ref{def:singlesystemsummoning}, the single-system summoning task assumes that only one of the diamonds is called, i.e., $b_i = 1$ for exactly one $D_i \in \mathcal{S}$. For an entanglement summoning task specified by a set of diamonds $\mathcal{D}$, several steps of our protocol involve ``summoning'' a system $A$ through a subset of diamonds $\mathcal{S}\subset \mathcal{D}$. This means that we define a single-system summoning task on $\mathcal{S}$ with input system $A$, using the bits $b_i$ that are given to $D_i \in \mathcal{S}$ for the entanglement summoning task as the input bits for the single-system summoning task. When the causal graph $G_{\mathcal{S}}$ is a tournament and when only one diamond in $\mathcal{S}$ is called, this is guaranteed to bring $A$ to the return point of the called diamond. In entanglement summoning, however, there are two input bits with value $1$ since two diamonds are called. If both called diamonds are in $\mathcal{S}$, it is not immediately clear how the single-system summoning protocol will behave. Fortunately, in the case where $G_{\mathcal{S}}$ is a tournament (which has no bidirected edges), the behaviour is quite simple, as given by the following proposition.
\begin{proposition}\label{prop:twocallsinglesystem}
Let $\mathcal{S}$ be a set of causal diamonds whose corresponding causal graph is a tournament. If calls are made to two diamonds $D_j, D_k \in \mathcal{S}$, with $D_j \to D_k$, executing a single-system summoning protocol on $\mathcal{S}$ returns the input system $A$ at the return point of $D_j$.
\end{proposition}
\begin{proof}
$\mathcal{S}$ satisfies Condition~(ii) of Theorem~\ref{thm:single}, so when there is a call only to $D_j$ (and not to $D_k$), single-system summoning successfully returns $A$ at the return point of $D_j$. Since $D_k \not\rightarrow D_j$, what is returned at $D_j$ cannot be causally influenced by whether $D_k$ is called. Therefore, when calls are made to both $D_j$ and $D_k$, $A$ must still be returned at $D_j$. 
\end{proof}

With these preliminaries in hand, we can describe our protocol for unassisted entanglement summoning.

\setcounter{protocol}{1}
\begin{protocol}~\label{UESprotocol}~\\
\noindent \ul{Preparation:}
\begin{itemize}
    \item For each pair of diamonds $\{D_j,D_{j'}\} \in \mathcal{X}_{\textrm{part}}$, prepare a Bell pair  $\ket{\Phi}_{A_{j}A_{j'}}$ in the distant past. 
    \item For each diamond $D_k \in \mathcal{D} \setminus \mathcal{X}$, prepare a Bell pair $\ket{\Phi}_{A_kB_k}$ in the distant past.
    \item Prepare the resources for performing a single-system summoning task on $\mathcal{S}_j$ for each $j \in \{1,\dots,n\}$.
\end{itemize}
\ul{Execution:}
\begin{itemize}
    \item At the call point of each $D_i \in \mathcal{D}$, broadcast the call information $b_i$ to the return point of every $D_j$ for which $D_i \to D_j$.
    \item For $\ket{\Phi}_{A_{j}A_{j'}}$ associated to each pair $\{D_j, D_{j'}\} \in \mathcal{X}_{\textrm{part}}$, summon $A_{j}$ through $\mathcal{S}_{j'}$ and $A_{j'}$ through $\mathcal{S}_{j}$.
    \item For $\ket{\Phi}_{A_kB_k}$ associated to each $D_k\in \mathcal{D} \setminus \mathcal{X}$, send $A_k$ to $D_k$ and summon $B_k$ through $\mathcal{S}_k$. 
    \item At the return point of each $D_j \in \mathcal{D}$: 
    \begin{itemize}
        \item If $D_j$ was called (i.e., $b_j = 1$), check the call information $b_i$ from every $D_i$ such that $D_i\rightarrow D_j$.
    \begin{itemize}
        \item If $b_i = 1$ for some $i$ such that $D_i\rightarrow D_j$, return $A_{i'}$ if $D_i \in \mathcal{X}$ (where $i'$ is such that $\{D_i,D_{i'}\} \in \mathcal{X}_{\textrm{part}}$) or return $B_{i}$ if $D_i \in \mathcal{D} \setminus \mathcal{X}$.\footnote{We show below that the required system will be available at $D_j$.}
        \item Else, if $b_i = 0$ for all $i$ such that $D_i\rightarrow D_j$, return $A_j$.
    \end{itemize}
    \end{itemize} 
\end{itemize}
\end{protocol}

Note that since the entanglement summoning task is assumed to satisfy the conditions of Theorem~\ref{thm:entanglementSummoningUnassisted}, the causal graph corresponding to $\mathcal{S}_j$ for every $j \in \{1,\dots, n\}$ is a tournament, so single-system summoning on $\mathcal{S}_j$ is achievable by Theorem~\ref{thm:single}. To verify this protocol, we must consider all possible choices of called diamonds. Let $D_{j}$ and $D_{k}$ be the two called diamonds, i.e., $b_{j} = b_{k} = 1$. Either $D_{j} \not\sim D_{k}$ or $D_{j} \sim D_{k}$. 

Consider first the case where $D_{j} \not\sim D_{k}$. This implies that $\{D_{j},D_{k}\} \in \mathcal{X}_{\textrm{part}}$, so a Bell pair $\ket{\Phi}_{A_{j}A_{k}}$ is prepared in the preparation phase. $A_{j}$ is summoned through $\mathcal{S}_{k}$ and $A_{k}$ through $\mathcal{S}_{j}$.  Recalling the definition of $\mathcal{S}_i$ from Eq.~\eqref{Sj}, we have that $D_{k} \in \mathcal{S}_{j}$ and $D_{j} \in \mathcal{S}_{k}$. Moreover, since $D_{j} \not\in \mathcal{S}_{j}$ and $D_{k}\not\in \mathcal{S}_{k}$, there is only one called diamond in each of the sets $\mathcal{S}_{j}$ and $\mathcal{S}_{k}$. Hence, summoning $A_{j}$ through $\mathcal{S}_{k}$ successfully brings $A_{j}$ to $D_{j}$, and summoning $A_{k}$ through $\mathcal{S}_{j}$ brings $A_{k}$ to $D_{k}$. Then, since $D_{j}$ sees that $b_i = 0$ for all $i$ such that $D_i \to D_{j}$, $D_{j}$ returns $A_{j}$ as per the protocol. Similarly, since $b_i = 0$ for all $i$ such that $D_i \to D_{k}$, $D_{k}$ returns $A_{k}$. This successfully completes the task, as the state on $A_{j}A_{k}$ is indeed $\ket{\Phi}$.

Next, consider the case where $D_{j} \sim D_{k}$. Suppose that $D_{j} \to D_{k}$ without loss of generality. We further divide this into two subcases, 1)~$D_{j} \in \mathcal{X}$ and 2)~$D_{j} \in \mathcal{D}\setminus \mathcal{X}$.
\begin{enumerate}[1)]
    \item $D_{j} \in \mathcal{X}$: Since $D_{j} \in \mathcal{X}$, there exists a (unique) $j' \in \{1,\dots, n\}$ such that $D_j \not\sim D_{j'}$, i.e., $\{D_j, D_{j'}\} \in \mathcal{X}_{\textrm{part}}$. The protocol therefore prepares a Bell pair $\ket{\Phi}_{A_j A_{j'}}$ during the preparation phase, and summons $A_j$ through $\mathcal{S}_{j'}$ and $A_{j'}$ through $\mathcal{S}_j$. $D_j \to D_k$ (together with the fact that $G_{\mathcal{D}}$ is an oriented graph) implies that $D_k \in \mathcal{S}_j$, and clearly, $D_j \not\in \mathcal{S}_j$. Hence, there is only one called diamond in $\mathcal{S}_j$, so summoning $A_{j'}$ through $\mathcal{S}_j$ successfully brings $A_{j'}$ to $D_k$. As for $\mathcal{S}_{j'}$, $D_j \in \mathcal{S}_{j'}$ by definition of $j'$, but it is possible that $k \in \mathcal{S}_{j'}$ as well. However, $D_j \to D_k$ by assumption, so by Proposition \ref{prop:twocallsinglesystem}, summoning $A_j$ through $\mathcal{S}_{j'}$ brings $A_j$ to $D_j$ even if $D_k \in \mathcal{S}_{j'}$. Then, since $D_j$ sees that $b_i = 0$ for all $D_i \to D_j$, $D_j$ returns $A_j$. Since $D_k$ sees that $b_j = 1$ and since $D_j \in \mathcal{X}$, $D_k$ returns $A_{j'}$.
    \item $D_j \in \mathcal{D} \setminus \mathcal{X}$: Since $D_j \in \mathcal{D} \setminus \mathcal{X}$, the protocol prepares a Bell pair $\ket{\Phi}_{A_jB_j}$ during the preparation phase. $A_j$ is sent to $D_j$ while $B_j$ is summoned through $\mathcal{S}_j$. $D_k \in \mathcal{S}_j$ and $D_j \not\in\mathcal{S}_j$ as in the previous subcase, so $D_k$ receives $B_j$. Since $D_j$ sees that $b_i = 0$ for all $D_i \to D_j$, $D_j$ returns $A_j$. Since $D_k$ sees that $b_j = 1$ and since $D_j \in \mathcal{D}\setminus \mathcal{X}$, $D_k$ returns $B_j$. 
\end{enumerate}
Thus, Protocol~\ref{UESprotocol} succeeds in every case, proving that any unassisted entanglement task satisfying the conditions of Theorem~\ref{thm:entanglementSummoningUnassisted} is achievable.

\subsection{Assistance through labelled calls}

We can ask whether giving Alice additional information makes the task easier. As with two-system summoning, we can consider labelling the calls to break the symmetry of the two calls. We saw in \S\ref{sec:LABS} that in the case two-system summoning, this type of assistance enlarges the set of achievable tasks (relative to the unassisted scenario). In entanglement summoning, on the other hand, assistance through labelled calls does not make any additional tasks achievable. 

For completeness, we define and characterize label-assisted entanglement summoning.
\begin{definition}
A \emph{label-assisted entanglement summoning task} is specified by a set of causal diamonds $\{D_j\}_{j=1}^n$, each of which is defined by a call point $c_j$ and return point $r_j$. The task involves two agencies, Alice and Bob, which perform the following.
\begin{enumerate}
    \item At each $c_j$, Bob gives Alice a number $b_j\in\{0,1,2\}$. Alice is promised that $b_{j^*} = 1$ for exactly one $j^*\in \{1,\dots, n\}$ and that $b_{k^*} = 2$ for exactly one $k^*\in \{1,\dots,n\}$, but does not know the values of $j^*$ and $k^*$ in advance.
\end{enumerate}
The task is \emph{achievable} if for any choice of $j^*,k^* \in \{1,\dots,n\}$ by Bob, Alice can return a system $A$ at $r_{j^*}$ and a system $B$ at $r_{k^*}$ such that the state on $AB$ is the Bell pair $\ket{\Phi}$.
\end{definition}

Restricting to one-way causal connections, the set of achievable label-assisted entanglement summoning tasks coincides with the set of achievable unassisted entanglement summoning tasks. Thus, if an entanglement summoning task is unachievable in the unassisted context, it remains unachievable even when we label the calls.
\begin{theorem}[label-assisted entanglement summoning]
A label-assisted entanglement summoning task whose causal graph $G_{\mathcal{D}}$ is an oriented graph is achievable if and only if $G_{\mathcal{D}}$ satisfies the conditions of Theorem~\ref{thm:entanglementSummoningUnassisted}.
\end{theorem}
\begin{proof}
The same argument as in the proof of Lemma~\ref{lemma:monogamy} gives that any label-assisted entanglement summoning task corresponding to the ``two-out graph'' in Fig.~\ref{fig:twoOut}, or any spanning subgraph thereof, is unachievable. It then follows directly from the proof of Theorem~\ref{thm:entanglementSummoningUnassisted} that the same necessary conditions applies.

\S\ref{sec:UESprotocol} showed that any unassisted entanglement summoning task whose causal graph satisfies the conditions of Theorem~\ref{thm:entanglementSummoningUnassisted} is achievable by providing a protocol. Clearly, any task that is achievable in the unassisted setting, where the calls are not labelled, is still achievable if labels are added to the calls. One can simply use the protocol for the unassisted task, ignoring the labels.
\end{proof}

\subsection{Assistance through global calls}

In a global-assisted entanglement summoning task, all diamonds except two receive $0$ at the call points, as in the previous scenarios, but the two called diamonds $D_{j^*}$ and $D_{k^*}$ receive the tuple $(j^*,k^*)$. That is, each of the two calls reveal the identity of \emph{both} diamonds which should return halves of a Bell pair. In this case, the entanglement summoning task becomes trivial---any task is achievable, regardless of the causal structure of the set of diamonds. To complete a task on an arbitrary set of diamonds $\mathcal{D}$, prepare a Bell pair $\ket{\Phi}_{A_{jk}B_{jk}}$ for every pair of diamonds $D_j, D_k \in \mathcal{D}$ with $j < k$, and send $A_{jk}$ to to $D_j$ and $B_{jk}$ to $D_k$. Then, at the return point of each $D_i \in \mathcal{D}$, if $D_i$ was called, i.e., if $i \in \{j^*,k^*\}$, returns the system indexed by $j^*k^*$; if $D_i$ was not called, return nothing.

\section{Tasks with bidirectional causal connections}\label{sec:bidirected}

In Sections~\ref{sec:twosystems} and~\ref{sec:entanglementSummoning}, we focused on two-system and entanglement summoning tasks corresponding to oriented graphs, i.e., graphs with no bidirected edges. For any pair of causal diamonds $D_j$ and $D_k$ in such a task, either $D_j \to D_k$ or $D_k \to D_j$, but not both. However, it is certainly geometrically possible for the causal connection between two diamonds to run both ways. Hence, it may be of interest to characterize tasks with bidirectional causal connections as well. These correspond to causal graphs that contain bidirected edges.

Note that the results of the previous sections in fact already characterize a large class of graphs with bidirected edges. If a two-system or entanglement summoning task is achievable on an oriented graph $G$, then the same type of task is achievable on \textit{any} graph $G'$ that can be obtained from $G$ by adding edges. The task on $G'$ can be completed by executing the protocol for $G$, ignoring all of the additional edges. Thus, we immediately know that any unassisted (resp.~label-assisted) two-system summoning task is achievable if its causal graph contains an oriented graph satisfying the conditions of Theorem~\ref{thm:unassistedtwosystemsummoning} (resp.~Theorem~\ref{thm:labelassistedtwosystemsummoning}) as a spanning subgraph. Likewise, unassisted or label-assisted entanglement summoning is achievable on any graph that contains an oriented graph satisfying the conditions of Theorem~\ref{thm:entanglementSummoningUnassisted} as a spanning subgraph.

Therefore, the graphs that remain to be studied are those whose oriented spanning subgraphs all correspond to unachievable tasks. These turn out to be of greater technical difficulty to characterize, as it is harder to generalize necessary conditions and protocols for small graphs to those with arbitrary numbers of vertices. Nevertheless, we present partial results for unassisted two-system and entanglement summoning.

\subsection{Two-system summoning with bidirectional causal connections}

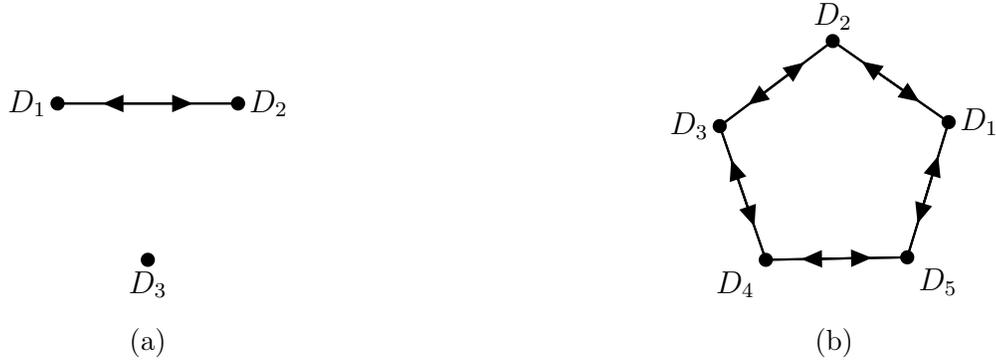
\begin{figure}
\centering
\begin{subfigure}[b]{.45\textwidth}
\centering
\begin{tikzpicture}[scale=0.4]
\begin{scope}[thick,decoration={
    markings,
    mark=at position 0.5 with {\arrow{triangle 45}}}
    ] 

\draw[fill=black] (0,0) circle (0.2cm);
\node[below] at (0,0) {$D_3$};
\draw[fill=black] (3,5.2) circle (0.2cm);
\node[right] at (3,5.2) {$D_2$};
\draw[fill=black] (-3,5.2) circle (0.2cm);
\node[left] at (-3,5.2) {$D_1$};

\draw[thick] (-3,5.2) -- (3,5.2);
\draw[thick,-triangle 45] (-3,5.2) -- (1.5,5.2);
\draw[thick,-triangle 45] (3,5.2) -- (-1.5,5.2);

\end{scope}
\end{tikzpicture}
\caption{\label{fig:bidirectionExampleOne}}
\end{subfigure}
\hfill
\begin{subfigure}[b]{.45\textwidth}
\centering
\begin{tikzpicture}[scale=1.6,rotate=19]
\begin{scope}[thick,decoration={
    markings,
    mark=at position 0.5 with {\arrow{triangle 45}}}
    ] 
  
\draw[fill=black] (1,0) circle (0.05cm);
\node[right] at (1,0) {$D_1$};

\draw[fill=black] (0.309,0.951) circle (0.05cm);
\node[above] at (0.309,0.951) {$D_2$};

\draw[fill=black] (-0.809,0.588) circle (0.05cm);
\node[left] at (-0.809,0.588) {$D_3$};

\draw[fill=black] (-0.809,-0.588) circle (0.05cm);
\node[below left] at (-0.809,-0.588) {$D_4$};

\draw[fill=black] (0.309,-0.951) circle (0.05cm);
\node[below right] at (0.309,-0.951) {$D_5$};

\begin{scope} [rotate=0]
\draw (1,0) -- (0.309,0.951);
\draw[-triangle 45] (1,0) -> (0.481,0.713);
\draw[-triangle 45] (0.309,0.951) -> (0.827,0.238);
\end{scope}

\begin{scope} [rotate=72]
\draw (1,0) -- (0.309,0.951);
\draw[-triangle 45] (1,0) -> (0.481,0.713);
\draw[-triangle 45] (0.309,0.951) -> (0.827,0.238);
\end{scope}

\begin{scope} [rotate=72*2]
\draw (1,0) -- (0.309,0.951);
\draw[-triangle 45] (1,0) -> (0.481,0.713);
\draw[-triangle 45] (0.309,0.951) -> (0.827,0.238);
\end{scope}

\begin{scope} [rotate=72*3]
\draw (1,0) -- (0.309,0.951);
\draw[-triangle 45] (1,0) -> (0.481,0.713);
\draw[-triangle 45] (0.309,0.951) -> (0.827,0.238);
\end{scope}

\begin{scope} [rotate=72*4]
\draw (1,0) -- (0.309,0.951);
\draw[-triangle 45] (1,0) -> (0.481,0.713);
\draw[-triangle 45] (0.309,0.951) -> (0.827,0.238);
\end{scope}

\end{scope}
\end{tikzpicture}
\caption{\label{fig:pentagon}}
\end{subfigure}
\caption{(a) A simple three-vertex graph with a bidirected edge, and two pairs of non-adjacent vertices. Entanglement summoning tasks with this causal graph are achievable, but two-system summoning tasks are not. (b) A five-vertex graph with bidirected edges. Two-system summoning is unachievable on this graph, as can be seen from the fact that the graph in (a) is a vertex-induced subgraph. Entanglement summoning tasks corresponding to this graph have not yet been characterized.\label{fig:bidirectedExamples}}
\end{figure}

We begin with unassisted two-system summoning. Recall from Fig.~\ref{fig:twosystemgraphs} that when we restrict to one-way causal connections, any task whose causal graph has more than one pair of non-adjacent vertices is unachievable. It is natural to ask if this remains true when we allow bidirectional causal connections. Consider the causal graph in Fig.~\ref{fig:bidirectionExampleOne}, where $D_3$ is not adjacent to $D_1$ nor $D_2$, so there are two pairs of non-adjacent vertices $D_3 \not\sim D_1$ and $D_3 \not\sim D_2$. If the causal connection between $D_1$ and $D_2$ were unidirectional, we know from Theorem~\ref{thm:unassistedtwosystemsummoning} that any corresponding unassisted two-system summoning task would be unachievable. Even with the bidirectional edge between $D_1$ and $D_3$, any task with this graph is still unachievable. This can be shown using the same type of argument as in the proofs of Lemmas~\ref{lemma:noTwoOutGraphs} and~\ref{lemma:nothreecycle}, to see that any successful protocol could be manipulated to clone one of the input subsystems. 

The unachievability of unassisted two-system summoning on the graph in Fig.~\ref{fig:bidirectionExampleOne} immediately leads to a necessary condition on tasks with an arbitrary number of diamonds. In any achievable task, each diamond can only be causally disconnected from at most one other diamond. In terms of the causal graph, each vertex $D_j$ can only be non-adjacent to one other vertex $D_{j'}$---otherwise, if $D_j \not\sim D_{j'}$ and $D_j \not\sim D_{j''}$ for $j' \neq j''$, the subgraph induced by $\{D_j,D_{j'},D_{j''}\}$ would be a subgraph of Fig.~\ref{fig:bidirectionExampleOne}. Thus, the subset of diamonds that are not connected to every other diamond can be partitioned into pairs $\{D_j, D_{j'}\}$ such that $D_j \not\sim D_{j'}$, but $D_j \sim D_k$ and $D_{j'} \sim D_k$ for all other diamonds $D_k$.

Now, consider the causal graph in Fig.~\ref{fig:doubleEdgeSquare}, in which there are two disjoint pairs of non-adjacent vertices. We show below that any unassisted two-system summoning task corresponding to this graph is achievable, which establishes that when bidirectional causal connections are allowed, tasks with more than one pair of causally disconnected diamonds can be achievable. This also serves as a counterexample to Theorem~5 of~\cite{adlam2018relativistic}, which claims that an unassisted two-system summoning task is achievable only if there is no more than one pair of disconnected diamonds.

To construct our protocol for unassisted two-system summoning on the graph in Fig.~\ref{fig:doubleEdgeSquare}, we first introduce a subroutine that we call a ``bounce.'' A ``bounce'' is executed on two diamonds $D_j$ and $D_k$ that share a bidirectional causal connection, $D_j \leftrightarrow D_k$, and involves some quantum system $X$, which starts at one of the diamonds. When $X$ starts at $D_j$, we denote the subroutine by $D_j \righttoleftarrow_X D_k$. The goal of this ``bounce'' subroutine is to bring $X$ to the return point $r_k$ of $D_k$ if $D_k$ is called, and to the return point $r_j$ of $D_j$ if $D_k$ is not called. Note that we cannot do this by sending $X$ directly to $r_k$; $r_k$ may not be in the causal past of $r_j$, so in the case where $D_k$ was not called, it may not be able to send $X$ back to $D_j$. (More generally, even if $D_j \leftrightarrow D_k$, it is possible that $D_j \cap D_k = \varnothing$, so there is no spacetime point that both has access to the call information at $D_j$ and $D_k$ \emph{and} from which $X$ can be forwarded to $D_j$ or $D_k$.)
Instead, we implement $D_j \righttoleftarrow_X D_k$ as follows.\\

\noindent\textbf{Subroutine: $D_j \righttoleftarrow_X D_k$}

\noindent \textit{\ul{Preparation:}
\begin{itemize}
    \item Prepare a Bell pair $\ket{\Phi}_{EE'}$. Send $E$ to the call point of $D_j$, and send $E'$ to the call point of $D_k$.
    \item Send $X$ to the call point of $D_j$.
\end{itemize}}

\noindent \textit{\ul{Execution:}
\begin{itemize}
    \item At $D_j$:
    \begin{itemize}
        \item At $c_j$, measure $XE$ in the Bell basis and broadcast the measurement outcome to $r_j$ and $r_k$.
        \item At $r_j$, if $E'$ is received from $D_k$, apply Pauli corrections to $E'$ based on the measurement outcome received from $c_j$.
    \end{itemize}
    \item At $D_k$:
    \begin{itemize}
        \item If a call is received at $c_k$, send $E'$ to $r_k$. At $r_k$, apply Pauli corrections to $E'$ based on the measurement outcome received from $c_j$.
        \item If no call is received at $c_k$, send $E'$ to $r_j$.
    \end{itemize}
\end{itemize}}

With the ``bounce'' subroutine defined, we can succinctly describe a protocol that completes any unassisted two-system summoning task corresponding to Fig.~\ref{fig:squareWithSystems}. At the start point $s$, encode the $A$ subsystem of the input state $\ket{\Psi}_{ABR}$ using a $(3,5)$ secret-sharing scheme, such that any three of the five shares can be used to recover $A$~\cite{cleve1999share}. Denote the five shares by $\alpha_1,\dots,\alpha_5$. Similarly, encode the $B$ subsystem into five shares $\beta_1,\dots,\beta_5$. Fig.~\ref{fig:squareWithSystems} shows where the shares should be initially sent in the preparation phase (e.g., $\alpha_1$, $\beta_1$, and $\alpha_5$ should be sent to the call point of $D_1$). Table~\ref{tab:squareprotocol} provides the instructions on which shares to then send to or ``bounce'' off of which diamonds, given each possible configuration of calls. It can be verified that in every case, the return point of one of the called diamonds receives three of the $\alpha$ shares, while the return point of the other called diamond receives three of the $\beta$ shares. Subsystems $A$ and $B$ can subsequently be recovered from these shares, and returned at the two called diamonds as required.

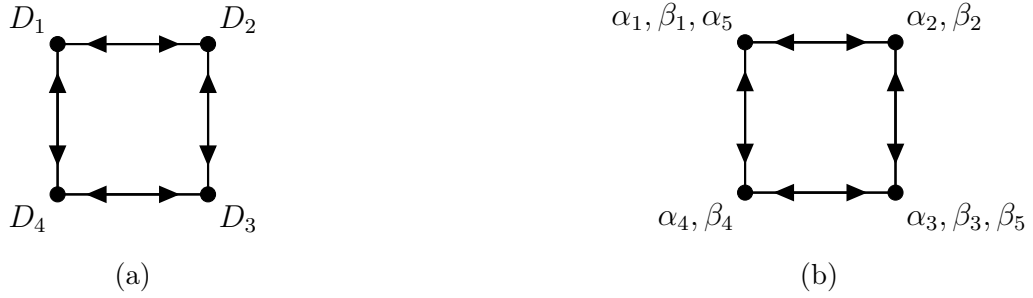
\begin{figure}
\centering
\begin{subfigure}[b]{.45\textwidth}
\centering
\begin{tikzpicture}[scale=0.5]

\draw[fill=black] (-2,-2) circle (0.2cm);
\node[below left] at (-2,-2) {$D_4$};

\draw[fill=black] (2,-2) circle (0.2cm);
\node[below right] at (2,-2) {$D_3$};

\draw[fill=black] (2,2) circle (0.2cm);
\node[above right] at (2,2) {$D_2$};

\draw[fill=black] (-2,2) circle (0.2cm);
\node[above left] at (-2,2) {$D_1$};

\draw[thick] (-2,-2) -- (2,-2);
\draw[thick,-triangle 45] (-2,-2) -- (1.25,-2);
\draw[thick,-triangle 45] (2,-2) -- (-1.25,-2);

\begin{scope}[rotate=90]
\draw[thick] (-2,-2) -- (2,-2);
\draw[thick,-triangle 45] (-2,-2) -- (1.25,-2);
\draw[thick,-triangle 45] (2,-2) -- (-1.25,-2);
\end{scope}

\begin{scope}[rotate=180]
\draw[thick] (-2,-2) -- (2,-2);
\draw[thick,-triangle 45] (-2,-2) -- (1.25,-2);
\draw[thick,-triangle 45] (2,-2) -- (-1.25,-2);
\end{scope}

\begin{scope}[rotate=270]
\draw[thick] (-2,-2) -- (2,-2);
\draw[thick,-triangle 45] (-2,-2) -- (1.25,-2);
\draw[thick,-triangle 45] (2,-2) -- (-1.25,-2);
\end{scope}

\end{tikzpicture}
\caption{\label{fig:doubleEdgeSquare}}
\end{subfigure}
\hfill
\begin{subfigure}[b]{.45\textwidth}
\centering
\begin{tikzpicture}[scale=0.5]
  
\draw[fill=black] (-2,-2) circle (0.2cm);
\node[below left] at (-2,-2) {$\alpha_4,\beta_4$};

\draw[fill=black] (2,-2) circle (0.2cm);
\node[below right] at (2,-2) {$\alpha_3,\beta_3,\beta_5$};

\draw[fill=black] (2,2) circle (0.2cm);
\node[above right] at (2,2) {$\alpha_2,\beta_2$};

\draw[fill=black] (-2,2) circle (0.2cm);
\node[above left] at (-2,2) {$\alpha_1,\beta_1,\alpha_5$};

\draw[thick] (-2,-2) -- (2,-2);
\draw[thick,-triangle 45] (-2,-2) -- (1.25,-2);
\draw[thick,-triangle 45] (2,-2) -- (-1.25,-2);

\begin{scope}[rotate=90]
\draw[thick] (-2,-2) -- (2,-2);
\draw[thick,-triangle 45] (-2,-2) -- (1.25,-2);
\draw[thick,-triangle 45] (2,-2) -- (-1.25,-2);
\end{scope}

\begin{scope}[rotate=180]
\draw[thick] (-2,-2) -- (2,-2);
\draw[thick,-triangle 45] (-2,-2) -- (1.25,-2);
\draw[thick,-triangle 45] (2,-2) -- (-1.25,-2);
\end{scope}

\begin{scope}[rotate=270]
\draw[thick] (-2,-2) -- (2,-2);
\draw[thick,-triangle 45] (-2,-2) -- (1.25,-2);
\draw[thick,-triangle 45] (2,-2) -- (-1.25,-2);
\end{scope}
\end{tikzpicture}
\caption{\label{fig:squareWithSystems}}
\end{subfigure}
\caption{(a) A four-vertex graph with two (disjoint) pairs of  non-adjacent vertices. We show that unassisted two-system summoning is achievable on this graph, contradicting~\cite{adlam2018relativistic}. (b) The initial distribution of shares used in our protocol. \label{fig:doublesquare}}
\end{figure}

\begin{table}
\centering
\begin{tabular}{|c|p{30mm}|p{30mm}|}
\hline
& if called & if not called \\
\hline
$D_1$ & $D_1 \righttoleftarrow_{\alpha_1} D_4$ \newline $D_1 \righttoleftarrow_{\beta_1} D_2$ \newline keep $\alpha_5$ & $D_1 \righttoleftarrow_{\alpha_1} D_2$ \newline $D_1 \righttoleftarrow_{\beta_1} D_4$ \newline send $\alpha_5$ to $D_2$ \\
\hline
$D_2$ & $D_2 \righttoleftarrow_{\alpha_2} D_1$ \newline $D_2 \righttoleftarrow_{\beta_2} D_3$ & $D_2 \righttoleftarrow_{\alpha_2} D_3$ \newline $D_2 \righttoleftarrow_{\beta_2} D_1$ \\
\hline
$D_3$ & $D_3 \righttoleftarrow_{\alpha_3} D_2$ \newline $D_3 \righttoleftarrow_{\beta_3} D_4$ \newline keep $\beta_5$ & $D_3 \righttoleftarrow_{\alpha_3} D_4$ \newline $D_3 \righttoleftarrow_{\beta_3} D_2$ \newline send $\beta_5$ to $D_4$ \\
\hline
$D_4$ & $D_4 \righttoleftarrow_{a_4} D_3$ \newline $D_4 \righttoleftarrow_{\beta_4} D_1$ & $D_4 \righttoleftarrow_{\alpha_4} D_1$ \newline $D_4 \righttoleftarrow_{\beta_4} D_3$ \\
\hline
\end{tabular}
\caption{Instructions for completing unassisted two-system summoning tasks corresponding to Fig.~\ref{fig:doubleEdgeSquare}. Beginning with the shares distributed as in Fig.~\ref{fig:squareWithSystems}, the table specifies what to do with each of the shares at the call point of each diamond, depending on whether that diamond receives a call.}
\label{tab:squareprotocol}
\end{table}

\subsection{Entanglement summoning with bidirectional causal connections}

As shown in \S~\ref{sec:UESprotocol}, in any achievable unassisted entanglement summoning task with no bidirectional causal connections, each diamond can only be causally disconnected from at most one other diamond. If we allow for bidirectional causal connections, however, this is no longer the case. Consider the example in Fig.~\ref{fig:bidirectionExampleOne}, in which $D_3$ is disconnected from both $D_1$ and $D_2$, and $D_1 \leftrightarrow D_2$. Any unassisted entanglement summoning task with this causal graph can be completed by pre-distributing a Bell pair between $D_1$ and $D_2$, and another Bell pair between $D_1$ and $D_3$.

Note that Fig.~\ref{fig:bidirectionExampleOne} still satisfies the necessary condition of Theorem~\ref{thm:entanglementSummoningUnassisted}. Recall that the necessity of this condition follows from the monogamy of entanglement. Even though Theorem~\ref{thm:entanglementSummoningUnassisted} was stated for oriented graphs, this necessary condition in fact applies to all causal graphs, regardless of whether there are bidirected edges. This can be seen from the fact that the proof in Theorem~\ref{thm:entanglementSummoningUnassisted} does not rely on the causal graph being an oriented graph. 

However, we do not know whether the condition of Theorem~\ref{thm:entanglementSummoningUnassisted} is still sufficient when we allow bidirectional causal connections. 
An intriguing example is provided by the five-vertex graph shown in Fig.~\ref{fig:pentagon}. This graph satisfies the condition of Theorem~\ref{thm:entanglementSummoningUnassisted}, so we have not identified any fundamental reason that the corresponding task should be unachievable. It can even be verified that unassisted entanglement summoning is achievable on any of its subgraphs induced by four vertices. Nonetheless, we have not been able to find a protocol which completes this task, so the achievability of this task remains unknown. We leave understanding this and entanglement summoning with bidirected edges more generally to future work.

\section{Discussion}\label{sec:discussion}

The quantum task framework is one way of understanding how quantum information may be moved through and processed in spacetime. In this article, we studied the achievability of two types of relativistic quantum tasks which formalize the problem of distributing two subsystems of a quantum state under timing constraints. We determined necessary and sufficient conditions for the achievability of these tasks in the case where the spacetime regions occupied by the parties do not have bidirectional causal connections. In proving our conditions, we introduced techniques that may be useful for analyzing and constructing protocols for more general quantum tasks. 

The two types of tasks we consider, two-system summoning and entanglement summoning, differ in the state to be distributed. In the former, the state is an unknown input, whereas in the latter, it is fixed to be a maximally entangled bipartite state. For each type of task, we considered the effect of various degrees of assistance provided in the form of additional information, and found changes in the conditions for achievability. The changes reflect coordination issues inherent to these summoning tasks, which certain types of additional information can help resolve. For instance, providing label assistance in two-system summoning allows Alice to avoid unintentionally summoning the two subsystems to the same diamond, while providing global assistance in entanglement summoning ensures that Alice never returns halves of different Bell pairs. 

The patterns that appear in analyzing these tasks suggest that there are theorems of greater generality yet to be proved. One such theorem~\cite{dolev2019constraining} was indeed inspired by the investigation in this paper. As another example, there is an intriguing similarity between the coordination problems encountered in single-system summoning with unrestricted calls \cite{adlam2018relativistic} and two-system summoning with unlabelled calls [cf.~\S\ref{sec:unassisted}]. In both cases, multiple identical inputs lead to tasks with cyclic causal graphs being unachievable. It would be interesting to understand if there is a general connection between indistinguishable inputs and causal graphs with topological ordering. We believe that a more precise statement could also be made about the power of entanglement in overcoming coordination problems arising from a lack of causal connections, as we have seen in a number of our protocols.  

\section*{Acknowledgements}

We thank Patrick Hayden for helpful discussions. KD is supported by the Center for Science of Information (CSoI), an NSF Science and Technology Center, under grant agreement CCF-0939370. AM acknowledges support from the It from Qubit Collaboration, which is sponsored by the Simons Foundation. AM was also supported by a CGS-D award given by the National Research Council of Canada. KW is supported by the Stanford Graduate Fellowship. This work benefited from a visit to the Stanford Institute for Theoretical Physics, supported by a Michael Smith Foreign Study Supplement given by the National Research Council of Canada. 

\appendix

\addtocontents{toc}{\fixappendix}

\section{Some graph lemmas}\label{appendix:graphlemmas}

In this appendix, we provide two technical results that we use to prove Theorems~\ref{thm:unassistedtwosystemsummoning} and~\ref{thm:labelassistedtwosystemsummoning}.
We say that a pair of vertices $D_j$ and $D_k$ in a directed graph $G = (V,E)$ are non-adjacent if $(D_j,D_k) \not\in E$ and $(D_k,D_j) \not\in E$, and write $D_j \not\sim D_k$ in this case. In the context of causal graphs, this corresponds to the causal diamonds $D_j$ and $D_k$  having no causal connection [cf.~Definition~\ref{def:causallyconnected}]. We also write $D_j \to D_k$ as shorthand for $(D_j,D_k) \in E$, and $D_j \not\rightarrow D_k$ for $(D_j, D_k) \not\in E$.

\begin{lemma} \label{lemma:onemissingedge}
Suppose that a directed graph $G = (V,E)$ satisfies the following conditions:
\begin{enumerate}
    \item $G$ is an oriented graph (i.e., $D_j \to D_k \Rightarrow D_k \not\rightarrow D_j$), and
    \item the subgraph of $G$ induced by any three vertices either has two edges pointing into one of the vertices, or is a cycle of length~$3$.
\end{enumerate}
Then, $D_j \not\sim D_k$ for at most one pair of vertices $D_j, D_k \in V$.
\end{lemma}
\begin{proof}
Suppose that there are two or more pairs of non-adjacent vertices. Condition~(ii) is clearly violated if any two pairs share a vertex, i.e., if $D_1 \not\sim D_2$ and $D_2\not\sim D_3$ for some vertices $D_1,D_2,D_3$. Hence, it must be the case that $D_1 \not\sim D_2$ and $D_3 \not\sim D_4$ for distinct vertices $D_1,D_2,D_3,D_4$. Applying Condition~(ii) to the subgraph induced by $\{D_1, D_2, D_3\}$, we see that we must have $D_1 \to D_3$ and $D_2 \to D_3$. Similarly, applying Condition~(ii) to the subgraph induced by $\{D_1, D_3, D_4\}$, we see that $D_3 \to D_1$ and $D_4 \to D_1$. But $D_1 \to D_3, D_3 \to D_1$ violates Condition~(i). Therefore, $G$ has at most one pair of non-adjacent vertices.
\end{proof}

If we strengthen Condition~(ii) of Lemma~\ref{lemma:onemissingedge} by disallowing cycles of length $3$, the graph is constrained to have one of the two forms in Fig.~\ref{fig:twosystemgraphs}, as shown below.

\begin{lemma}\label{lemma:unassistedgraphs}
Suppose that a directed graph $G = (V,E)$ satisfies the following conditions:
\begin{enumerate}
    \item $G$ is an oriented graph (i.e., $D_j \to D_k \Rightarrow D_k \not\rightarrow D_j$), and
    \item in the subgraph of $G$ induced by any three vertices, there are two edges pointing into one of the vertices.
\end{enumerate}
Then, $G$ is either (a) a transitive tournament, or (b) obtained from a transitive tournament by removing the edge between the first two vertices in its topological ordering [cf.~Fig.~\ref{fig:twosystemgraphs}].
\end{lemma}
\begin{proof}
By Lemma \ref{lemma:onemissingedge}, there is at most one pair of non-adjacent vertices in $G$. Note that Condition~(ii) further implies that $G$ cannot contain any cycles of length~3. Thus, in the case where $G$ has no pairs of non-adjacent vertices, it is a tournament with no cycles, i.e., a transitive tournament. Likew, in the case where $G$ has one pair of non-adjacent vertices, say $D_1 \not\sim D_2$, the subgraph induced by $V \setminus \{D_1,D_2\}$ must also be a transitive tournament. Moreover, applying Condition~(ii) to the subgraph induced by $\{D_1, D_2, D_j\}$ implies that $D_1 \to D_j$ and $D_2 \to D_j$ for every $D_j \in V \setminus \{D_1,D_2\}$, so $G$ is the same as a transitive tournament in which $D_1$ and $D_2$ are the first two vertices in its topological ordering, except with the edge between $D_1$ and $D_2$ removed.
\end{proof}

\section{Protocols for assisted two-system summoning}

In this appendix, we describe protocols for all label-assisted two-system summoning tasks that satisfy the conditions of Theorem~\ref{thm:labelassistedtwosystemsummoning}. We also construct a global-assisted two-system summoning protocol for the causal graph in Fig.~\ref{fig:twomissingpossible}. The protocols are described for the case where the input subsystems $A$ and $B$ are qubit systems; the generalization to larger systems is trivial.

\subsection{Protocol for label-assisted two-system summoning}\label{appendix:labelassistedtwosystem}

Consider a label-assisted two-system summoning task satisfying the conditions of Theorem~\ref{thm:labelassistedtwosystemsummoning}. Lemma~\ref{lemma:onemissingedge} implies that the causal graph corresponding to any such task contains either no non-adjacent vertices or one pair of non-adjacent vertices. In this section, we present protocols for both cases, thereby proving that the conditions of Theorem~\ref{thm:labelassistedtwosystemsummoning} are sufficient for the task to be achievable.

First suppose that there are no non-adjacent vertices in the causal graph, i.e., $D_j \sim D_k$ for every pair of diamonds $D_j$ and $D_k$ [cf.~Fig.~\ref{fig:tournament}]. In this case, the task can be reduced to two instances of single-system summoning, since every pair of diamonds is causally connected (satisfying the conditons of Theorem~\ref{thm:single}) and the calls are distinguishable. It can be completed by applying the single-system summoning protocol twice in parallel, summoning the $A$ subsystem to the call labelled $1$ and the $B$ subsystem to the call labelled $2$ (or vice versa).

In the second case where there is a pair of non-adjacent vertices in the causal graph, this strategy is no longer viable. We cannot use single-system summoning on the full set of diamonds, as one pair of diamonds is not causally connected (in violation of Theorem~\ref{thm:single}). Instead, we provide the following protocol (Protocol~\ref{LBSprotocol}). Here, the two causally disconnected diamonds are labelled $D_0$ and $D_0'$ (so $D_0 \not\sim D_0'$), and the remaining diamonds are labelled $D_1, \dots, D_{n-2}$. Note that by Condition~(ii) of Theorem~\ref{thm:labelassistedtwosystemsummoning}, we must have $D_0 \to D_j$ and $D_0' \to D_j$ for every $j \in \{1,\dots, n-2\}$ [cf.~Fig.~\ref{fig:rivalGraphb}]. 

\begin{protocol} \label{LBSprotocol}~\\
\noindent
\ul{Preparation:}
\begin{itemize}
    \item Prepare Bell pairs $\ket{\Phi}_{EE'}$ and $\ket{\Phi}_{FF'}$. Send $E$ and $F$ to the call point of $D_0$, and send $E'$ and $F'$ to the call point of $D_0'$.
    \item Send subsystem $A$ (from the start point $s$) to the call point of $D_0$, and subsystem $B$ to the call point of $D_0'$.
    \item At the call point of $D_0$, prepare ancillary systems $G_1$ and $G_2$. At the call point of $D_1$, prepare ancillary systems $G_1'$ and $G_2'$.
    \item Prepare the resources for performing two single-system summoning tasks on $\{D_1,\dots, D_{n-2}\}$ with the call point of $D_0$ as the start point. One task responds only to a call labelled $1$, and the other responds only to a call labelled $2$. 
    \item Prepare the resources for performing two single-system summoning tasks on $\{D_1,\dots, D_{n-2}\}$ with the call point of $D_0'$ as the start point. One task responds only to a call labelled $1$, and the other responds only to a call labelled $2$. 
\end{itemize}
\ul{Execution:}
\begin{itemize}
    \item At $D_0$:
    \begin{itemize}
        \item At the call point $c_0$, broadcast the call information $b_0$ to every $r_j$ with $j \geq 1$.
        \item If a call is received at $c_0$, i.e., if $b_0 \in \{1,2\}$, return $A$ at the return point $r_0$. Additionally:
        \begin{itemize}
            \item If $b_0 = 1$, swap $F$ onto $G_2$.
            \item If $b_0 = 2$, swap $F$ onto $G_1$.
        \end{itemize} 
        \item If no call is received at $c_0$, i.e., if $b_0 = 0$, measure $AE$ in the Bell basis and broadcast the measurement outcome to every $r_j$ with $j \geq 1$. Swap $F$ onto $G_1$.
    \end{itemize}
    \item At $D_0'$:
    \begin{itemize}
        \item At the call point $c_0'$, broadcast the call information $b_0'$ to every $r_j$ with $j \geq 1$.
        \item If a call is received at $c_0'$, i.e., if $b_0' \in \{1,2\}$, return $B$ at the return point $r_0'$. Additionally:
        \begin{itemize}
            \item If $b_0' = 1$, swap $E'$ onto $G_2'$.
            \item If $b_0' = 2$, swap $E'$ onto $G_1'$.
        \end{itemize} 
        \item If no call is received at $c_0'$, i.e., if $b_0' = 0$, measure $BF'$ in the Bell basis and broadcast the measurement outcome to every $r_j$ with $j \geq 1$. Swap $E'$ onto $G_2'$.
    \end{itemize}
    \item Summon $G_1$ and $G_1'$ (from the call points of $D_0$ and $D_0'$, respectively) through $\{D_1,\dots,D_{n-2}\}$ to the call labelled $1$. Summon $G_2$ and $G_2'$ through $\{D_1,\dots,D_{n-2}\}$ to the call labelled $2$.
    \item At the return point $r_j$ of $D_j$ for each $j \in \{1,\dots, n-2\}$:
    \begin{itemize}
        \item If a call was received at $c_j$, i.e., if $b_j = \ell$ for $\ell \in \{1,2\}$, check the call information from $D_0$ and $D_0'$.
        \begin{itemize}
            \item If $D_0$ was called, return\footnote{See footnote~\ref{Paulicorrectionfootnote}.} $G_\ell$.
            \item If $D_0'$ was called, return $G_\ell'$.
            \item If neither $D_0$ nor $D_0'$ were called, return $G_1$ if $\ell=1$ or $G_2'$ if $\ell=2$
        \end{itemize}
    \end{itemize}        
\end{itemize}
\end{protocol}

To see that this protocol works, consider what happens in each case. If $D_0$ and $D_0'$ receive the two calls, they return $A$ and $B$, respectively, completing the task. If neither $D_0$ nor $D_0'$ are called, subsystem $A$ (modulo Pauli corrections) is teleported to $c_0$ and swapped onto $G_1$, which is subsequently summoned to the call labelled $1$, while subsystem $B$ is teleported to $c_0'$ and swapped onto $G_2'$, which is subsequently summoned to the call labelled $2$. The two called diamonds see (at their return points) that $D_0$ and $D_0'$ were not called, and therefore return $G_1$ and $G_2'$ (after applying the Pauli corrections, which they also receive from the call points of $D_0$ and $D_0'$), as required. If $D_0$ is called but $D_0'$ is not, $D_0$ returns $A$, while $B$ (modulo Pauli corrections) is moved to $G_1$ or $G_2$ depending on the label of the call received at $D_0$. Specifically, if $D_0$ received $b_0 = 1$, it swaps $B$ onto $G_2$ to ensure that it is summoned to the diamond $D_j$ that received the \emph{other} label $b_j = 2$ (and vice versa). $D_j$ sees that $D_0$ was called, and returns $G_2$ (after applying the Pauli corrections received from $D_0'$). The protocol likew succeeds in the analogous case where $D_0'$ is called and $D_0$ is not. 

\subsection{Example protocol for global-assisted summoning}\label{appendix:globalassistedexample}

In this section, we give a protocol that completes any global-assisted two-system summoning task corresponding to the four-vertex causal graph in Fig.~\ref{fig:twomissingpossible}. Recall that in the global-assisted setting, if a diamond $D_i$ is called, it receives a tuple $(j^*,k^*)$ containing the indices of both called diamonds (with $i \in \{j^*,k^*\}$). Hence, in the case where $D_i$ is called, it also knows the identity of the other called diamond.

\begin{protocol}~\\
\noindent
\ul{Preparation:}
\begin{itemize} 
    \item Prepare two Bell pairs $\ket{\Phi}_{EE'}$ and $\ket{\Phi}_{FF'}$. Send $E$ and $F$ to the call point of $D_2$, and send $E'$ and $F'$ to the call point of $D_4$.
    \item Send subsystem $A$ (from the start point $s$) to the call point of $D_2$, and subsystem $B$ to the call point of $D_4$.
\end{itemize}
\ul{Execution:}
\begin{itemize}
    \item At $D_2$:
    \begin{itemize}
        \item If a call is received at $c_2$, return $A$ at $r_2$. Also, if the other call is to $D_\ell$ with $\ell \in \{1,3\}$, send $E$ to (the return point of) $D_\ell$. 
        \item If no call is received at $c_2$, send $E$ to $D_1$. Measure $AF$ in the Bell basis and broadcast the measurement outcome to $D_1$ and $D_3$.
    \end{itemize}   
    \item At $D_4$,
    \begin{itemize}
        \item If a call is received at $c_4$, return $B$ at $r_4$. Also, if the other call is to $D_\ell$ with $\ell \in \{1,3\}$, send $F'$ to $D_\ell$. 
        \item If no call is received at $c_4$, send $F'$ to $D_3$. Measure $BE'$ in the Bell basis and broadcast the measurement outcome to $D_1$ and $D_3$.
    \end{itemize}  
    \item At $D_1$:
    \begin{itemize}
        \item If a call is received at $c_1$:
        \begin{itemize}
            \item If the other call is to $D_2$ or $D_3$, return\footnote{See footnote~\ref{Paulicorrectionfootnote}.} $E$ at $r_1$.
            \item If the other call is to $D_4$, return $F'$ at $r_1$.
        \end{itemize}
    \end{itemize}
    \item At $D_3$, 
    \begin{itemize}
        \item If a call is received at $c_3$:
        \begin{itemize}
            \item If the other call is to $D_1$ or $D_4$, return $F'$ at $r_3$. 
            \item If the other call is to $D_2$, return $E$ at $r_3$.
        \end{itemize}
    \end{itemize}
\end{itemize}
\end{protocol}
One can check case by case that this protocol successfully completes the task. To understand the basic idea, observe that the subsystems $A$ and $B$ start at $D_2$ and $D_4$. Hence, if $D_2$ is called, it returns $A$, but if $D_2$ is not called, it does not know where to send $A$, since only called diamonds receive the tuple $(j^*,k^*)$. In this case, the strategy is to teleport (without the Pauli corrections) $A$ onto $F'$ at $D_4$. Then, if $D_4$ is called, it knows which of $D_1$ and $D_3$ is called, so it knows where to forward $F'$. Since $D_2 \to D_1, D_3$, both $D_1$ and $D_3$ receive the measurement outcomes from the teleportation, and can therefore apply the appropriate Pauli corrections to $F'$. On the other hand, if $D_4$ is not called, it does not know which of $D_1$ or $D_3$ is called, but always forwards $F'$ to $D_3$. Analogously, $D_2$ always forwards $E$ to $D_1$ when $D_2$ is not called. This ensures that when neither $D_2$ nor $D_4$ are called, $D_1$ and $D_3$ both get shares of the state. 

\section{Localizing bipartite states across spacetime regions}\label{sec:localizingentanglement}

In~\cite{hayden2016summoning,hayden2019localizing}, it was argued that single-system summoning has a simple physical interpretation: a quantum system $A$ can be localized to a set of causal diamonds if and only if the single-system summoning task specified by that set is achievable. $A$ is localized to a set of diamonds if without any prior knowledge of $A$, it is possible to collect quantum and classical systems from any one diamond in the set and use them to recover $A$. Theorem~\ref{thm:single} then becomes a statement about when a quantum system can be localized to every member of a set of spacetime regions, at least when those regions are causal diamonds. This was generalized to regions of arbitrary shape in~\cite{hayden2019localizing}.

We can ask whether there is a summoning task on a set of diamonds whose achievability implies that two subsystems $A$ and $B$ of a state $\ket{\Psi}_{ABR}$ can be localized across every pair of diamonds in the set. There are two ways one might define localizing $A$ and $B$ across a pair of diamonds $D_j$ and $D_k$: either we require $A$ and $B$ be recoverable from $D_j\cup D_k$, or we require that $A$ be recoverable from $D_j$ and $B$ from $D_k$, or vice versa. The first definition is the same as localizing the system $AB$ to the region $D_j\cup D_k$, and this is already characterized by~\cite{hayden2019localizing}. Hence, we use the second definition here. It is easy to characterize the sets of diamonds for which localizing two systems across every pair of diamonds is possible.

Suppose that it is possible to localize $A$ and $B$ across every pair of diamonds in $\mathcal{D}$. Let $\mathcal{A}$ (resp.~$\mathcal{B}$) be the subset of diamonds to which $A$ (resp.~$B$) can be localized. By the result of~\cite{hayden2016spacetime}, every pair of diamonds in $\mathcal{A}$ must be causally connected, and likewise for $\mathcal{B}$. Additionally, $\mathcal{A} \setminus \mathcal{B}$ can consist of at most one diamond. Otherwise, if two diamonds $D_j$ and $D_k$ are in $\mathcal{A} \setminus \mathcal{B}$, the only subsystem that can be recovered at $D_j$ and $D_k$ is $B$, so it would not be possible to localize $A$ and $B$ to every pair. Similarly, $\mathcal{B} \setminus \mathcal{A}$ consists of at most one diamond. It follows that there is at most one pair of diamonds in $\mathcal{D}$ that is causally disconnected. 

We can therefore localize $A$ and $B$ across every pair of diamonds by performing two single-summoning tasks.
Specifically, perform a single-system summoning task on $\mathcal{A}$ with input $A$, and a single-system summoning task on diamonds $\mathcal{B}$ with input $B$. This does not correspond to any of the two-system summoning tasks defined in Section~\ref{sec:twosystems}. 

\section*{References}

\bibliographystyle{unsrt}
\bibliography{biblio}

\end{document}